\documentclass{article}
\ifx\pdfoutput\undefined
\else
  \pdfoutput=1 
\fi
\usepackage{graphicx}
\usepackage{amsmath,amssymb,verbatim,amsthm}
\usepackage{subfigure}
\newtheorem{theorem}{Theorem}[section]
\newtheorem{lemma}[theorem]{Lemma}
\newtheorem{proposition}[theorem]{Proposition}
\newtheorem{corollary}[theorem]{Corollary}

\newtheorem{remark}[theorem]{Remark}

\numberwithin{equation}{section}

\usepackage{color}
\usepackage{xcolor}

\usepackage[textsize=tiny]{todonotes}

\newcommand{\bbE}{\mathbb{E}}

\newcommand{\bbR}{\mathbb{R}}

\newcommand{\bbK}{\mathbb{K}}

\newcommand{\cG}{\mathcal{G}}

\newcommand{\cC}{\mathcal{C}}
\newcommand{\cO}{\mathcal{O}}
\newcommand{\cL}{\mathcal{L}}

\newcounter{hypA}
\newenvironment{hypA}{\refstepcounter{hypA}\begin{itemize}
  \item[({\bf A\arabic{hypA}})]}{\end{itemize}}

\title{Multilevel Sequential Monte Carlo  
with Dimension-Independent Likelihood-Informed Proposals}

\author{
 A. Beskos, \thanks{
 Department of Statistical Science,  University College London, London, UK
  }
\and
A. Jasra, \thanks{
Department of Statistics \& Applied Probability, National University of Singapore, SG, 117546
  }
 \newline
  \and
K. J. H. Law, \thanks{
  Computer Science and Mathematics Division, Oak Ridge National Laboratory,
  Oak Ridge, TN, USA, 37831
  }
  \and
Y. M. Marzouk, \thanks{
  Department of Aeronautics \& Astronautics,
  Massachusetts Institute of Technology,
  Cambridge, MA, USA, 02139
  }
\and
Y. Zhou \thanks{
  Department of Statistics \& Applied Probability, National University of Singapore, SG, 117546
  }
}

\begin{document}
\maketitle

\begin{abstract}

In this article we develop a new sequential Monte Carlo (SMC) method
for multilevel (ML) Monte Carlo estimation.
In particular, the method can be used
to estimate expectations with respect to a target
probability distribution over an infinite-dimensional
and non-compact space---as given, for example, by a Bayesian inverse problem with Gaussian
random field prior.  Under suitable assumptions the MLSMC method has
the optimal $\cO(\varepsilon^{-2})$
bound on the cost to obtain a mean-square error of $\cO(\varepsilon^2)$.
The algorithm is accelerated by dimension-independent likelihood-informed (DILI)
proposals \cite{cui2014dimension} designed for Gaussian priors, leveraging a novel variation which uses
empirical sample covariance information in lieu of Hessian information,
hence eliminating the requirement for gradient evaluations.
The efficiency of the algorithm is illustrated on two examples: inversion of
noisy pressure measurements in a PDE model of Darcy flow 
to recover the posterior distribution of the permeability field,
and inversion of noisy measurements of the solution
of an SDE to recover the posterior path measure.

\end{abstract}

\section{Introduction}
\label{sec:intro}

The estimation of expectations with respect to a probability
distribution over an infinite-dimensional and non-compact space, and
of the normalizing constants of such distributions, has a wide range
of applications; see for instance \cite{stuart2010inverse} and the
references therein.  In particular, Bayesian inverse problems (BIP)
with Gaussian random field priors are an important class of such
mathematical models.  In most cases of practical interest, one must
compute these estimates using the Monte Carlo method under a
finite-dimensional discretization of the associated probability
distribution; see \cite{cotter, hoang}, for example.

In many scenarios, such as the BIP above, the finite-dimensional
approximation of the probability distribution of interest becomes more
accurate but more computationally expensive as the dimension of the
approximation goes to infinity. This is precisely the class of
problems which are of interest in this paper. It is well known that
the multilevel Monte Carlo (MLMC) method \cite{gile:08,hein:98} can
reduce the computational effort, relative to independent sampling,
required to obtain a particular mean square error; see \cite{hoang,
  ourmlsmc} for examples in the inverse problem
context. 
The MLMC idea introduces a sequence of 
increasingly accurate
  approximations of the target probability distribution, 
and relies on
independently sampling from a collection of couples of this sequence
and employing the multilevel (ML) identity; details are given later in
this paper. The main challenge in problems of interest here is
that such independent sampling is seldom possible.

This paper employs sequential Monte Carlo (SMC) samplers, as these
approaches have been shown to outperform Markov chain Monte Carlo
(MCMC) in many cases (e.g.,~\cite{jasra}) and to be robust in classes
of high-dimensional problems (see \cite{beskos,andrew,jasra17weak}).
In \cite{ourmlsmc} an SMC method for multilevel estimation was
introduced and analyzed for a class of BIPs. This method was developed
specifically for scenarios where ML estimation is expected to be quite beneficial,
but where independent sampling from the couplings of interest is
not trivial to perform. This method was extended to the estimation of normalizing constants in \cite{mlsmcnc}.
Both \cite{ourmlsmc,mlsmcnc} use SMC and importance sampling to replace independent sampling and
coupling in the multilevel context.
However, the approaches in \cite{ourmlsmc,mlsmcnc}
can only approximate a sequence of probability distributions on a
\textit{fixed} state space. In other words, the dimension of the
parameter of interest, and hence the state space of the resulting
Markov chain, is assumed to be fixed and finite. Levels in the
estimation scheme correspond to refinements of the forward model
PDE approximation.
This paper, in contrast, assumes that the parameter of interest is in
principle infinite-dimensional; thus the resolution of the parameter
is refined along with the approximation of the PDE model as the level
increases. The dimension of the state space of the resulting
Markov chain therefore increases at each level, and hence a modification of previous
multilevel algorithms in \cite{ourmlsmc,mlsmcnc} is required.
%

The contributions of this paper are as follows:
\begin{enumerate}
\item{The design of a new SMC sampler approach for
    MLMC estimation which includes refinement of the parameter space
    as well as of the forward model.} 
\item{Theoretical cost analysis for this general MLSMC method.} 
\item{Introduction of a covariance-based version of the likelihood-informed subspace of
\cite{cui2014likelihood, cui2014dimension} (abbreviated cLIS), and a method for its sample approximation.}
\item{Adoption of efficient dimension-independent likelihood-informed (DILI)
proposals \cite{cui2014dimension} within the SMC algorithm, utilizing the new cLIS.} 
\end{enumerate}
Note that SMC samplers that are robust in robust in high-dimensional
settings (see, e.g., \cite{beskos,andrew,jasra17weak}) rely on MCMC as
well as on sequential importance sampling/resampling. For such samplers
to work well, the MCMC step must be efficient, i.e.,~mix over the
high-dimensional state space at a reasonable rate. 
In the present context, we show that this mixing can be achieved through efficient DILI
proposals from \cite{cui2014dimension}.

This article is structured as follows. In Section \ref{sec:setup}, the
basic algorithm and estimation procedure are introduced. Section
\ref{sec:theory} presents some theoretical results for the
algorithm. Section \ref{sec:mlsmcdili} shows how the DILI proposal
methodology can be used in the context of MLSMC. Section
\ref{sec:numerics} presents several numerical implementations of our
methodology. Some technical mathematical results are deferred to the
appendix.

\section{Model and Approach} 
\label{sec:setup}

\subsection{Model}
\label{ssec:model}

Let $U_0,U_1, \dots$ be a sequence of spaces with $U_n\subseteq\mathbb{R}^{d'_n}$,
$d'_n\in\mathbb{N}$ and $n\geq 0$.
Let $E_n = \bigotimes_{i=0}^n U_i \subseteq \bbR^{d_n}$, where
$d_n=\sum_{i=0}^{n} d'_i$. 
We consider a sequence of probability measures $\{\hat{\eta}_n\}_{n\geq 0}$ on spaces $\{E_n\}_{n\geq 0}$.
We denote the densities w.r.t.~an appropriate dominating measure as $\{\hat{\eta}_n\}_{n\geq 0}$.
We suppose that
$$
\hat{\eta}_n(u_{0:n}) = \frac{\kappa_n(u_{0:n})}{Z_n}
$$
with $\kappa:E_n\rightarrow\mathbb{R}^+$ known but $Z_n$ possibly unknown.
In practice, these probability measures
are associated with a Bayesian inverse problem and in particular a (basis function-type)
approximate solution of a partial differential equation.
As $n$ grows, so does the dimension of the target, but to a well defined infinite-dimensional limit.
Let the approximate solution of the continuous system 
associated to an input $u_{0:\ell}\in E_\ell$ processed into
a finite number $p\in\mathbb{N}$ of summary values be denoted $\rho_\ell$,
i.e.,~$\rho_\ell:E_\ell\rightarrow\mathbb{R}^{p}$.
We are interested in computing, for bounded-measurable functions
$\varphi:\mathbb{R}^p\rightarrow\mathbb{R}$
$$
\hat{\eta}_L(\varphi\circ\rho_L) := \int_{E_L} \varphi(\rho_L(u_{0:L})) \hat{\eta}_L(u_{0:L})du_{0:L}
$$
for some large $L$ or preferably $\hat{\eta}_{\infty}(\varphi\circ\rho_\infty)$.
Denote the infinite resolution target by
$\eta(\varphi) := \hat{\eta}_{\infty}(\varphi\circ\rho_\infty)$.
In addition it is of interest to estimate $Z_L$ or $Z_{\infty}$.
Define $\rho_l(u_{0:n}) := \rho_l(u_{0:l})$ for $n>l$.

Assume that
\begin{equation}
\label{eq:seq_kappa}
\kappa_\ell(du_{0:\ell}) = \cL(\cG_\ell(u_{0:\ell})) \mu_0(du_{0:\ell}),
\end{equation}
where $\cL$ is a likelihood term, $\cG_\ell:E_\ell \rightarrow \bbR^q$
is the map from parameter input $u_{0:\ell}\in E_\ell$
to $q\in\mathbb{N}$ observations of the
approximate solution of the continuous system, 
and $\mu_0$ is the 
prior density,
where the limiting prior measure $\mu_0$ is defined on $E_\infty$,
and the density of its finite-dimensional distribution is taken for
$u_{0:\ell} \in E_\ell$.
It is worth noting that the theory, to be described later on, will be more broadly
applicable than the context described in this paragrapgh.


\subsection{Algorithm and Estimator}

We consider a sequence of Markov kernels $\{K_n\}_{n\geq 0}$
$K_n:E_n\rightarrow\mathcal{P}(E_n)$ ($\mathcal{P}(E_n)$ are the probability measures on $E_n$)
which each keep the respective $\{\hat{\eta}_n\}_{n\geq 0}$ invariant, i.e.,~$\hat{\eta}_n K_n = \hat{\eta}_n$.
Let $\{q_n\}_{n\geq 1}$ be a sequence of probability kernels on $\{U_n\}_{n\geq 1}$, $q_n:E_{n-1}\rightarrow\mathcal{P}(U_n)$.
Let $\{M_n\}_{n\geq 1}$, $M_n:E_{n-1}\rightarrow\mathcal{P}(E_n)$ be defined as
\begin{equation*}
M_{n}(u_{0:n-1},du_{0:n}') = K_{n-1}(u_{0:n-1},du_{0:n-1}')\otimes q_n(u_{0:n-1}',du'_n).
\end{equation*}
Finally, let
\begin{equation*}
G_0(u_0) = 1
\end{equation*}
and for $n\geq 1$
\begin{equation*}
G_n(u_{0:n}) = \frac{\kappa_n(u_{0:n})}{\kappa_{n-1}(u_{0:n-1})q_n(u_{0:n-1},u_n)},
\end{equation*}
where slightly degenerate notation has been used for $q_n(u_{0:n-1},du_n)=q_n(u_{0:n-1},u_n)du_n$.
For $n\geq 0$, $\varphi\in\mathcal{B}_b(E_n)$ set
\begin{equation*}
\gamma_n(\varphi) := \int_{E_0\times\cdots\times E_n }\varphi(u_{0:n}(n))\prod_{p=0}^{n-1} G_p(u_{0:p}(p))
\hat{\eta}_0(du_0(0)) \prod_{p=1}^{n} M_p(u_{0:p-1}(p-1),du_{0:p}(p))
\, ;
\end{equation*}
then one can show that, for $n\geq 1$, $\hat{\eta}_n(\varphi) = \gamma_n(G_n\varphi)/\gamma_n(G_n)$. Denote
$\eta_n(\varphi) = \gamma_n(\varphi)/\gamma_n(1)$. We note that $Z_n/Z_0 = \gamma_n(G_n)$. We set $\eta_0 = \hat{\eta}_0$.

Let $n\geq 1$, $\mu\in\mathcal{P}(E_{n-1})$ and define $\Phi_n:\mathcal{P}(E_{n-1})\rightarrow \mathcal{P}(E_{n})$
$$
\Phi_n(\mu)(du_{0:n}) = \frac{\mu(G_{n-1}M_n(\cdot,du_{0:n}))}{\mu(G_{n-1})}.
$$
Our multilevel algorithm works as follows. Let $N_0\geq N_1\geq \cdots\geq N_L\geq 1$ be a sequence of integers that are given.
The algorithm approximates the sequence $\{\eta_n\}_{L\geq n\geq 0}$. At time zero, one samples
$$
\prod_{i=1}^{N_0} \eta_0(du_0^i(0)) \, .
$$
Let $\eta_0^{N_0}$ denote the $N_0-$empirical measure of samples.
At time 1, one samples from
$$
\prod_{i=1}^{N_1} \Phi_1(\eta_0^{N_0})(du_{0:1}^i(1)) \, .
$$
Thus, in an obvious extension of the notation, the joint law of the algorithm is
$$
\Big(\prod_{i=1}^{N_0} \eta_0(du_0^i(0))\Big)\Big(\prod_{\ell=1}^L \prod_{i=1}^{N_l} \Phi_\ell(\eta_{\ell-1}^{N_{\ell-1}})(du_{0:\ell}^i(\ell))\Big).
$$

Notice that the present algorithm is different from the one in
\cite{delm:06b} and hence the one in 
\cite{ourmlsmc}. 
 In particular, the state space dimension grows at each iteration. 
%
Also the algorithm will have increasing cost with time (the subscript
$n$). This is because the cost of the MCMC steps and sometimes the
cost of computing the $G_n$ will grow at some rate with the size of
the state space. This is generally not desirable for classical
applications of SMC methods associated with the filtering of
non-Gaussian and nonlinear state-space models (i.e.,~dynamic problems
with data arriving sequentially in time).  However, the algorithm
described above is designed for inverse problems that are so-called
``static,'' i.e.,~one has a single instance of the data from which to
make inference. Such growth is therefore less of a concern.

\subsubsection{Estimation}\label{ssec:est_smc}

We note that
$$
\hat{\eta}_L(\varphi\circ\rho_L) = \sum_{\ell=0}^L [\hat{\eta}_\ell(\varphi\circ\rho_\ell) - \hat{\eta}_{\ell-1}(\varphi\circ\rho_{\ell-1})]
$$
with $\hat{\eta}_{-1}(\varphi\circ\rho_{-1}):=0$.
Now
$$
\hat{\eta}_\ell(\varphi\circ\rho_\ell)  = \frac{Z_{\ell-1}}{Z_\ell}\hat{\eta}_{\ell-1}M_\ell(G_\ell\varphi\circ\rho_\ell) = \frac{Z_{\ell-1}}{Z_\ell}\left(\hat{\eta}_{\ell-1}\otimes q_\ell\right )(G_\ell\varphi\circ\rho_\ell).
$$
Also $\eta_\ell (\varphi \circ \rho_{\ell-1}) = \hat \eta_{\ell-1} (\varphi \circ \rho_{\ell-1})$.
So one can approximate, for $\ell\geq 1$, $\hat{\eta}_\ell(\varphi\circ\rho_\ell) - \hat{\eta}_{\ell-1}(\varphi\circ\rho_{\ell-1})$ by
$$
\eta_{\ell}^{N_\ell}(G_\ell)^{-1}\eta_{\ell}^{N_\ell}(G_\ell\varphi\circ\rho_\ell) - \eta_{\ell}^{N_\ell}(\varphi\circ\rho_{\ell-1}),
$$
and $\hat{\eta}_0(\varphi\circ\rho_0)$ by $\eta^{N_0}_0(\varphi\circ\rho_0)$. This estimate is \emph{different} than that in \cite{ourmlsmc},
but similar in spirit. As $Z_\ell/Z_0 = \gamma_\ell(G_\ell)$, this can be approximated by
$$
\gamma_{\ell}^{N_{0:\ell}}(G_{\ell}) = \prod_{l=0}^\ell \eta_{l}^{N_l}(G_l).
$$
As shown in \cite{mlsmcnc} (in a different context) this estimator has similar properties to one that follows the `standard' ML type principle.

\subsection{No Discretization Bias}

We can also consider removing the bias of the estimators in Section \ref{ssec:est_smc} using ideas from \cite{rg:15}; indeed this is achieved in \cite{vollmer}.
In particular, let $L\in\{0,1,2,\dots\}$ be a random variable that is independent of the MLSMC
algorithm with $\mathbb{P}_L(L\geq l)>0~\forall l\geq 0$. 
Then 
the following is an unbiased estimator of $\gamma_{\infty}(G_{\infty})=Z_{\infty}/Z_0$
$$
\sum_{\ell=0}^L
\frac{1}{\mathbb{P}_L(L\geq \ell)}
 \Big(\gamma_{\ell}^{N_{0:\ell}}(G_{\ell}) - \gamma_{\ell-1}^{N_{0:\ell-1}}(G_{\ell-1})\Big).
$$ 
The main barrier to show that this estimator is unbiased is to show that
\begin{equation}
\lim_{n\rightarrow\infty}\mathbb{E}[(\gamma_{n}^{N_{0:n}}(G_{n})-\gamma_{\infty}(G_{\infty}))^2] = 0\label{eq:ub1}.
\end{equation}
Following from the unbiased property of each $\gamma_{n}^{N_{0:n}}(G_{n})$ 
(\cite{delm:04}), we have
$$
\mathbb{E}[(\gamma_{n}^{N_{0:n}}(G_{n})-\gamma_{\infty}(G_{\infty}))^2] =
\mathbb{E}[\gamma_{n}^{N_{0:n}}(G_{n})^2] +
\gamma_{\infty}(G_{\infty})^2 - 2\gamma_{\infty}(G_{\infty})\gamma_{n}(G_{n}).
$$
Now if one ensures $N_n \geq n$, then one expects (via \cite[Theorem 1.1]{berard}) in the inverse problem context that
$$
\frac{\gamma_{n}^{N_{0:n}}(G_{n})^2}{\gamma_{\infty}(G_{\infty})^2} \rightarrow_{\mathbb{P}} 1.
$$
Therefore, under appropriate uniform integrability conditions,
$$
\lim_{n\rightarrow\infty} \mathbb{E}[\gamma_{n}^{N_{0:n}}(G_{n})^2] =
\gamma_{\infty}(G_{\infty})^2 
$$
and thus \eqref{eq:ub1} holds.
A similar remark can be used when estimating $\hat{\eta}_{\infty}$ except that the required asymptotics in $n$ do not appear to
be in the literature. 
Note that there is a large random cost for this method and 
hence we do not consider such an
approach further.

\section{Theoretical results}\label{sec:theory}

The following assumptions will be made. Throughout $E_\ell$ is compact.

\begin{hypA}
\label{hyp:A}
Assume there exist some $\underline{c}$, $\overline{C}$
such that for all $\ell=0,1,\dots, $ and all $u_{0:\ell} \in E_\ell$
\begin{equation}\label{eq:gbound}
0 < \underline{c} \leq G_\ell(u_{0:\ell}) \leq \overline{C} < \infty.
\end{equation}
\end{hypA}

\begin{hypA}
\label{hyp:B}
Assume there exists a $\lambda < 1$ such that for all $\ell=0,1,\dots$, 
$u,v \in E_\ell$ and $A\subset E_\ell$
$$
K_\ell(u,A) \geq \lambda K_\ell(v,A).
$$
\end{hypA}


\begin{hypA}
\label{hyp:C}
Assume there exists a $c>0$ and $\beta>0$ such that for all $\ell$
sufficiently large
\begin{equation}\label{eq:gee}
V_\ell := \max \{\|G_\ell - 1\|^2_\infty, \|\rho_\ell-\rho_{\ell-1}\|^2_\infty\} \leq c h_\ell^{\beta},
\end{equation}
where the sup-norm is with respect to the probability space.
Also, assume the cost $\cC_\ell$
to evaluate $G_\ell$ and $\rho_\ell$
satisfies, for some $\zeta\ge 0$,
$$\cC_\ell \leq c \, h_\ell^{-\zeta}.$$
\end{hypA}

Define
\[
\hat \eta^{\rm ML}_L(\varphi) := \eta_{0}^{N_0}(\varphi\circ\rho_{0}) + \sum_{\ell=0}^L  \eta_{\ell}^{N_\ell}(G_\ell)^{-1}\eta_{\ell}^{N_\ell}(G_\ell\varphi\circ\rho_\ell) - \eta_{\ell}^{N_\ell}(\varphi\circ\rho_{\ell-1}).
\]
Let $a(\epsilon) \lesssim b(\epsilon)$ denote that there exists a $c>0$
such that $a(\epsilon) \lesssim c b(\epsilon)$ for all $\epsilon$ sufficiently small.

\begin{proposition}
Assume (A\ref{hyp:A}-\ref{hyp:C}) and assume $\beta>\zeta$.
Then, for any $\varepsilon>0$
there exists an $L$, and a choice of $\{N_\ell\}_{\ell=0}^L$, such that
\begin{equation}
\bbE | \hat \eta^{\rm ML}_L(\varphi) - \hat \eta(\varphi) |^2 \lesssim \varepsilon^2,
\end{equation}
for a total cost ${\rm Cost} \lesssim \varepsilon^{-2}$.
\end{proposition}

\begin{proof}
The proof follows essentially that of \cite{ourmlsmc} given the above assumptions.
Observe that Lemma \ref{lem:main} (in the appendix) below provides the bound
\begin{align*}
\mathbb{E}\big[\{\hat \eta^{\rm ML}_L(g)-\mathbb{E}_{\eta_L}[g(U)]\}^2\big]
\leq
C\,\bigg(\frac{1}{N_0} + &\sum_{\ell=1}^{L}\frac{V_\ell}{N_{\ell}} 
+ \sum_{1\le \ell<q\le L}V_\ell^{1/2} V_q^{1/2}
\big\{\tfrac{\kappa^{q}}{N_{\ell}}
+\tfrac{1}{N_{\ell}^{1/2}N_{q}}
\big\} \bigg)\ .
\end{align*}
Theorem 3.3 of \cite{jasra2016forward}
describes how to complete the proof.
Briefly, the choice $L \eqsim |\log \varepsilon|$ controls the bias.
One chooses $N_\ell=\varepsilon^{-2}K_L h_\ell^{(\beta+\zeta)/2}$,
where $K_L = \sum_{\ell=1}^{L-1} h_\ell^{(\beta-\zeta)/2}=\cO(1)$,
so one has
$$
{\rm COST} = \sum_{\ell=0}^L N_\ell C_\ell = \varepsilon^{-2} K_L^{2}
\lesssim \varepsilon^{-2}.
$$
It then suffices to show the second term is negligible for this choice,
and this is done in Theorem 3.3 of \cite{jasra2016forward}.
\end{proof}

The below result follows directly from that in \cite{mlsmcnc} and hence the proof is omitted.

\begin{corollary}
Assume (A\ref{hyp:A}-\ref{hyp:C}) and assume $\beta>\zeta$.
Then, for any $\varepsilon>0$
there exists an $L$ and a choice of $\{N_\ell\}_{\ell=0}^L$ such that
\begin{equation}
\bbE | \gamma_L^{N_{0:L}}(G_L) - Z_\infty/Z_0 |^2 \lesssim \varepsilon^2,
\end{equation}
for a total cost ${\rm Cost} \lesssim |\log\varepsilon| \varepsilon^{-2}$.
\end{corollary}

%

\section{Dimension-Independent Likelihood-Informed  Proposals}
\label{sec:mlsmcdili}

In this section, we describe how to effectively embed the DILI 
proposals from \cite{cui2014dimension}
into the MLSMC framework.  
In particular, Section~\ref{sec:inf} 
describes the non-intrusive (i.e., gradient-free) covariance-based construction
of the cLIS. 
Section \ref{sec:em} illustrates an approach for obtaining the dimension of the cLIS.
Section \ref{sec:dili} provides a description of the resulting DILI proposals
which will ultimately be used as the kernels $K_\ell$ 
within the MLSMC algorithm and, finally,
Section~\ref{sec:mlis}
places the cLIS construction within the multilevel context. Later, we will present simulation studies that illustrate the significance of such likelihood-informed proposal for the effectiveness of the overall MLSMC method.

\subsection{Sample Approximation to cLIS} 
\label{sec:inf}
For simplicity of exposition, this section will consider (high)
finite dimension $d$; however the framework is easily extended to infinite dimensional spaces.
Consider the case where we have a particle population  $u^i \in \bbR^d$, $1\le i\le N$,
for some $N\ge 1$, from the probability measure $\eta\in\mathcal{P}(\bbR^{d})$,
for some (high) dimension $d\ge 1$.
Define the covariance operator
\begin{equation*}
C := \bbE_\eta [ (u - \bbE_\eta(u)) \otimes  (u - \bbE_\eta(u)) ],
\end{equation*}
and assume that
\begin{equation}
\label{eq:idea}
C = I - QQ^\top,
\end{equation}
where $Q \in \bbR^{d \times m}$, $I$ is the $d\times d$ identity matrix,
and $m \ll d$ is the dimension of 
a linear subspace of \emph{concentration}  of the measure $\eta$,
as characterized by the covariance operator,
with respect to some reference measure
$\eta_0$ such that
\begin{equation*}
\bbE_{\eta_0} (u - \bbE_{\eta_0}(u)) \otimes  (u - \bbE_{\eta_0}(u)) = I.
\end{equation*}
One should think of $\eta_0$ and $\eta$ as prior and posterior measures, respectively, in a given context.
The column space of $Q$ is a covariance-based generalization of
the gradient-based 
LIS introduced in \cite{cui2014likelihood, cui2014dimension},
and will be referred to below as a cLIS.
Notice that the condition above is equivalent to
\begin{equation*}
C^{-1} = I + M M^\top,
\end{equation*}
where $M$ and $Q$ have the same column space,
and if $\sigma^2_{Q,i}$, $\sigma^2_{M,i}$ are the squared singular values (i.e., squared eigenvalues of
the matrix times its transpose) of matrices $Q$, $M$, respectively, then
$\sigma_{M,i}^2 = \sigma_{Q,i}^{2}/(1- \sigma_{Q,i}^{2})$,
for $1\le i \le m$, once the appropriate ordering is applied.

It is often the case that in practice (\ref{eq:idea}) holds only approximately,
in the sense that
$C \approx I - QQ^\top$ (see \cite{cui2014dimension} for applications);
however for simplicity of the presentation here we will assume that it
holds exactly. It is worth noting that for multimodal posterior distributions, the
posterior covariance might not be a negative-definite perturbation of
the prior, as in (\ref{eq:idea}); in such cases the perturbation could
turn out to be indefinite or positive.
Indeed, for multi-modal targets, Gaussian proposals might be ineffective
and 
a mixture approach could be required.
This important issue is beyond the scope of the present work and is
not considered further.
%
%

We want to estimate $C$  and, more importantly,
the column space of $Q$, using the particles $\{ u^i \}_{i=1}^N$.
The simplest way this can be done is the following. 
Assume for simplicity that we know the rank  $m$.
We construct a sample approximation of the low-rank correction to the
covariance as
 \begin{equation}
 H_N := I - \frac1{N-1} \sum_{i=1}^N (u^i - \bar u) \otimes (u^i - \bar u),
 \qquad \bar u = \frac1N\sum_{i=1}^N u^i \, .
\label{eq:samplecorrect}
 \end{equation}
Now, we use an iterative algorithm, such as the Lanczos iteration, to compute the dominant $m$ eigenpairs 
giving rise to
$P_{N,m}\in \bbR^{d\times m}$, and a diagonal $\Lambda_{N,m} \in \bbR^{m\times m}$
(with the diagonal comprised of the $m$ dominant eigenvalues)
so that
$H_N \approx P_{N,m} \Lambda_{N,m} (P_{N,m})^\top$.
The (orthonormal) columns of $P_{N,m}$ correspond to the
$N$-sample approximation of the $m$-dimensional cLIS.
Simulations indicate that as long as $N>d$, this approach provides a reasonable approximation of the cLIS.
Indeed, \eqref{eq:idea} may be seen as an inverse version of the {\it spiked covariance model} from
\cite{donoho2013optimal}.
There it was shown that this is in fact the required number of samples
as $d\rightarrow \infty$, and explicit error bounds are provided.
See also \cite{zahm2016covariance} for further exploration of this
point.

For a simple example,
see Figure \ref{fig:lis_converge} where we consider 20 random Gaussian targets
with $d=100$ of known low rank $m=10$,
i.e., with a 10-dimensional cLIS.
These random Gaussian targets were constructed as follows.
For $k=1,\dots, 20$, $i=1,\dots, d$, and $j=1,\dots m$, let $A_{ij}^{(k)} \sim N(0,1)$, independently over $i,j$.
Let $C^{(k)}=(A^{(k)}A^{(k),\top} + I_d)^{-1}$.
The $k^{th}$ Gaussian is given by $N(0,C^{(k)})$, $k=1,\ldots, 20$.
The cLIS is approximated using \eqref{eq:samplecorrect} with i.i.d.\@ samples (Figure \ref{fig:lis_converge}, left panel) and highly correlated samples (Figure \ref{fig:lis_converge}, right panel),
and the cLIS fidelity is approximated using
\begin{equation}\label{eq:liserror}
{\rm fidelity} = \| PP^\top ( I - P_{N,m}P_{N,m}^{\top} ) \|,
\end{equation}
where $P$ is the matrix whose orthonormal columns comprise the exact $m$-dimensional
cLIS (analytically known in this synthetic example) and $\|\cdot\|$ indicates the Frobenius norm.
The rationale behind this non-symmetric subspace divergence is that we are particularly concerned
with how well $P_{N,m}$ approximates $P$, i.e., with the projection of $P_{N,m}$ onto $P$.
Note that a weighted subspace distance \cite{ieeepeople} as advocated
in \cite{cui2014likelihood}
can be used to favor recovery of the
most important directions of the cLIS; alternatively one might use a
modification of the F{\"o}rstner \cite{forstner2003metric} metric
between SPD matrices, as proposed in
\cite{cui2014dimension}. 
Also, note that these ideal error metrics cannot be computed in practice, since
we do not have access to $P$.  

\begin{figure}[h]
\centering
\includegraphics[width=0.48\columnwidth]{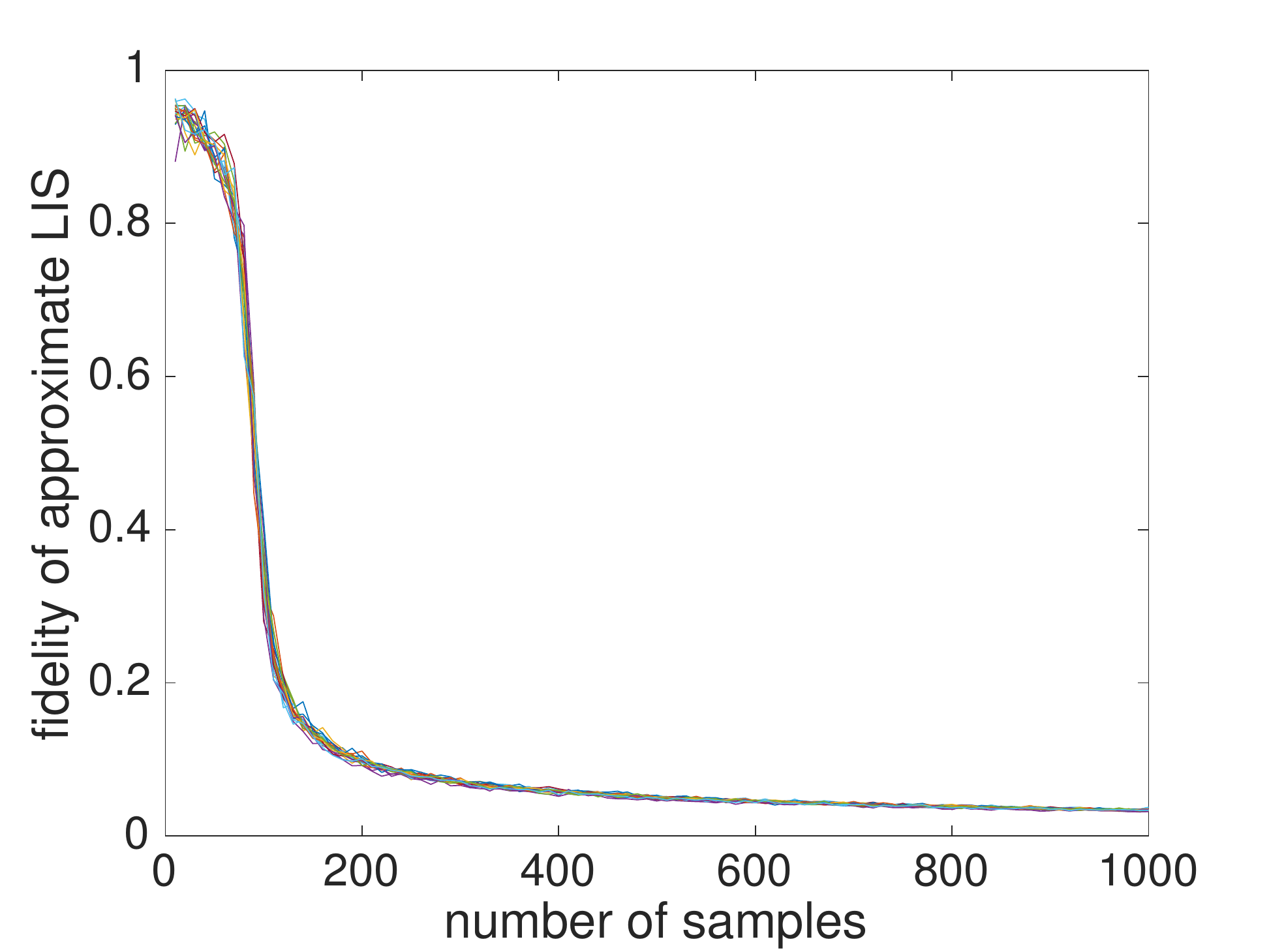}\ ,
\includegraphics[width=0.48\columnwidth]{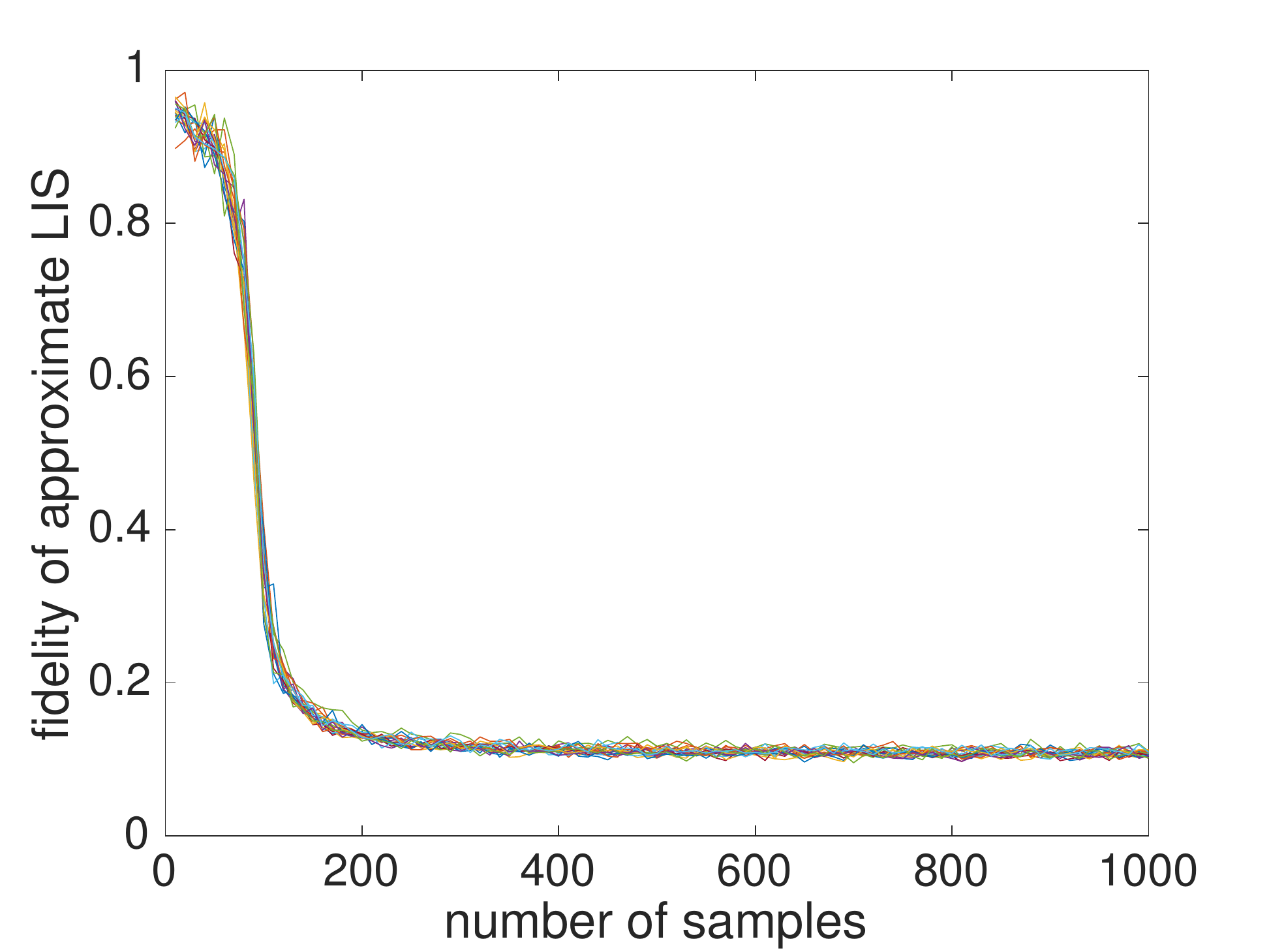}
\caption{Left panel: cLIS fidelity as a function of number of samples
  for 20  random targets, using i.i.d.\@  samples. Right panel: cLIS
  fidelity as a function of number of samples for random targets,
  using correlated samples with IACT=100. This example illustrates the
  fast convergence of the sample estimate of the cLIS in both cases.}
\label{fig:lis_converge}
\end{figure}


\begin{remark}
In general, the cLIS construction may miss local features
that can be
captured by the original gradient-based LIS of \cite{cui2014likelihood}.
However, the cLIS will ultimately be used here only for construction of a \emph{Gaussian} proposal,
and it is unclear what benefit a more sensitive gradient-based LIS
would offer for this purpose.
As a simple example, consider a $d$-dimensional measure which is bimodal
in one data-informed direction, e.g., with density proportional to
$\exp\{-\frac1{2\gamma^2} ((y_1-u_1^2)^2 + (y_2 - u_2)^2 ) - \frac12 |u|^2\}$.
As $\gamma \rightarrow 0$, the averaged Hessian used to build an LIS in
\cite{cui2014dimension}  will be large across $(u_1,u_2)$,
which both are clearly informed by the data.
The cLIS, though, will only identify $u_2$.
Nonetheless, a global Gaussian proposal constructed using
\textit{either} subspace will have difficulty sampling the target.
\end{remark}

\subsection{Estimating the Dimension of the cLIS}
\label{sec:em}

It is critical to develop a method to automatically estimate the cLIS dimension $m$ in realistic scenarios,
where one may know or suspect that there exists a low-dimensional subspace informed by the data of some unknown dimension $m\ge 1$.
 For this task, the following algorithm is proposed.

Let $\tilde h_N $ denote the full vector of $d$ eigenvalues of matrix $H_N$ in (\ref{eq:samplecorrect}),
sorted in decreasing order, and let $h_N = \tilde h_N \mathbf{1}_{\{\tilde h_N\geq 0\}}$.
{We truncate the negative eigenvalues, as there may be some large negative
eigenvalues when the sample covariance approximation is poor,
while the cLIS approximation can actually already be adequate.
This also prevents issues arising when a perturbation from the prior is not
negative definite,
as might occur with multi-modal posteriors.
Now define, for $i=1,\dots d-1$,
\begin{equation*}
(\widetilde{\Delta h_N})_i = |h_{N,i+1}-h_{N,i}|, \quad
\overline{\Delta h_N} = \frac{1}{d-1} \sum_{i=1}^{d-1}(\widetilde{\Delta h_N})_i,
\quad (\Delta h_N)_i = (\widetilde{\Delta h_N})_i/(\overline{\Delta h_N})_i.
\end{equation*}
It will suffice to find the
index $i_{\rm ex}$ such that $(\Delta h_N)_{i_{\rm ex}} >{\rm TOL}$,
{where ${\rm TOL}$ is some pre-specified reasonable value in between the sample error
and the expected size of the gap in the spectrum at convergence.}
This index,
effectively the position where the relative absolute difference in the
eigenvalues delivers a `spike,'
is then taken as the estimate of $\hat m = i_{\rm ex}$.

Figure \ref{fig:find_m} applies this approach to the target $N(0,C^{(1)})$, where $C^{(1)}$ is constructed as described above \eqref{eq:liserror}, i.e., one of the targets from
Figure~\ref{fig:lis_converge}.
The left panel of Figure \ref{fig:find_m} illustrates the growth of the gap in spectrum beyond the threshold for the target $N(0,C^{(1)})$
from Figure \ref{fig:lis_converge},
where the horizontal axis indicates the index of the non-zero values of $h_N$, and the vertical axis shows $\Delta h_N$ for $N=d, d+m, \dots, 250$.
The threshold value is set to ${\rm TOL}=10$, and $(\Delta h_N)_{i_{\rm ex}}$ exceeds this value already for $N=250$ samples,
with the correct value of $m=i_{\rm ex}=10$.
It is clear that in this case
the increments $\Delta h_N$ show a spike at  at the correct value $m=10$, for large enough sample-size.
Notice also that the right panel in Figure \ref{fig:find_m} illustrates that, in this example,
the sample approximation of cLIS is less challenging than the sample
approximation of the covariance matrix. 

\begin{figure}[!h]
\centering
\subfigure[Sequence of increments in vector  $\Delta h_N$,
for various choice of sample size $N$, for a given
target distribution $N(0,C^{(1)})$.]{
\includegraphics[width=0.48\columnwidth]{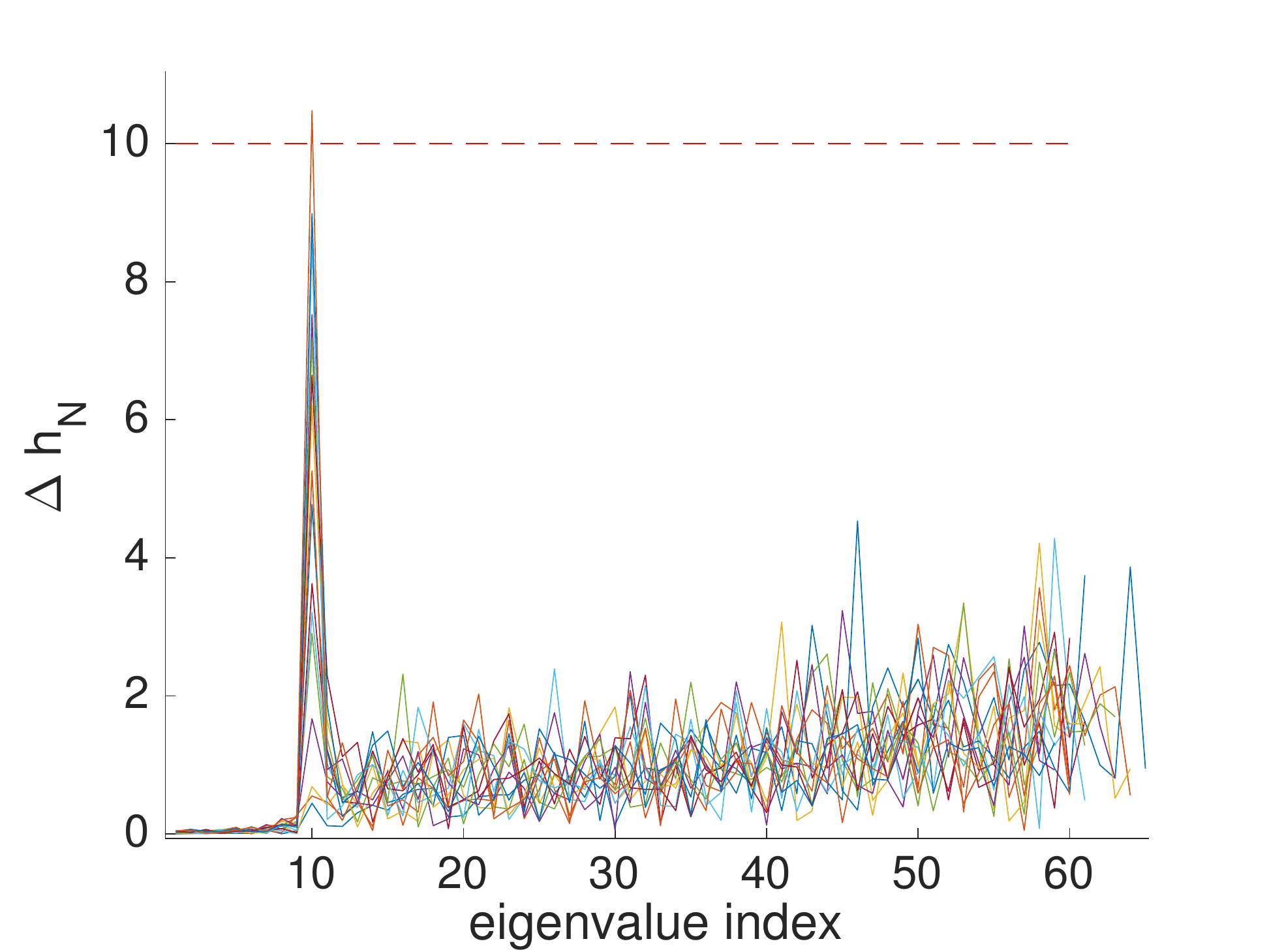}} 
\hspace{5pt}\subfigure[Relative fidelity of the cLIS approximation in comparison to relative fidelity of the sample covariance, denoted by $S$, in terms of number of samples.]
{\includegraphics[width=0.45\columnwidth]{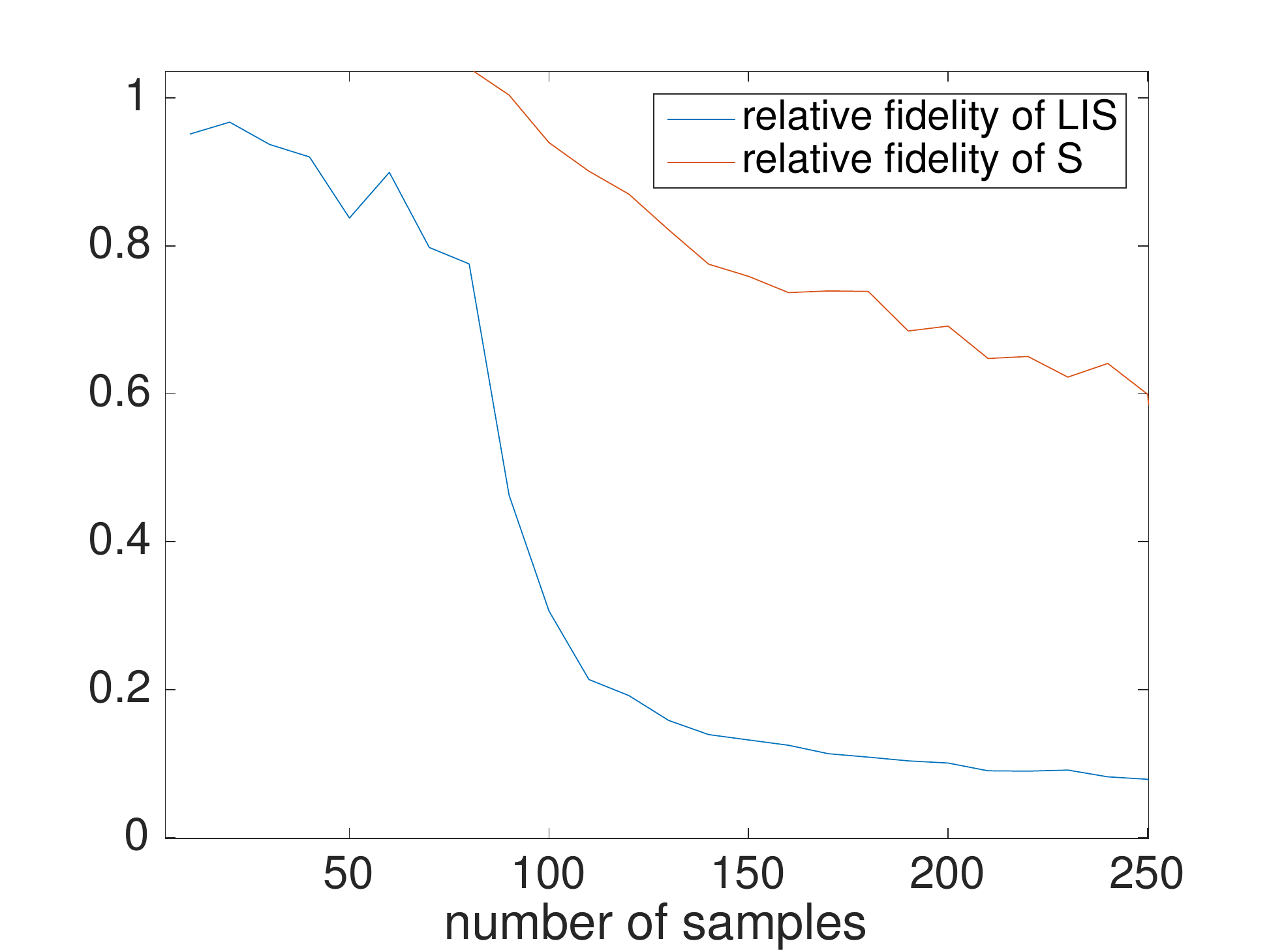}}
\caption{Convergence of approximate cLIS.}
\label{fig:find_m}
\end{figure}

\subsection{Use of a Subspace at Mutation Step}
\label{sec:dili}

Mutation steps in our SMC algorithm will use a  
DILI proposal,
defined abstractly as follows.
Consider a 
subspace determined by the collection of orthonormal vectors $P= [e_1, e_2, \ldots, e_m]$,
spanning an $m$-dimensional subspace of $\mathbb{R}^{d}$,
together with  an approximation of the mean $\bar u \approx \bbE_\eta u$
and the covariance of the coordinates $(\langle u, e_i \rangle)_{i=1}^{m} $,
$\Sigma\approx  P^\top C P \in \bbR^{m\times m}$.
We will make use of the orthogonal decomposition
\begin{equation*}
u  =  PP^{\top}u  + (I-PP^{\top})u
\end{equation*}
where $PP^{\top}u$ is the orthogonal projection of $u$ on the subspace.
Let $u' \sim Q(u,\cdot)$ be defined by
\begin{equation}
 u' = \bar u + A  (u -\bar u) + B w, \quad w\sim N(0,I),
 \label{eq:diliprop}
 \end{equation}
where we have defined
\begin{align}
A & =  P (I_m - b_m \Sigma)^{1/2} P^\top + (1- b_\perp^2)^{1/2} (I-PP^\top), \label{eq:pro}\\
B & =  P \sqrt{b_m \Sigma} P^\top + b_\perp(I-PP^\top),
\label{eq:props}
\end{align}
and $b_m, b_\perp \in (0,1)$ are small step sizes on and off the subspace, respectively.
The second summands in (\ref{eq:pro})--(\ref{eq:props}) correspond to a preconditioned Crank-Nicolson (pCN)
step on the space orthogonal to the subspace, while the first summands
correspond to a step that uses
the 
covariance $\Sigma$ to scale the step sizes across the various directions of the subspace.
All matrix operations are carried out via the eigendecomposition
of the symmetric, positive semi-definite $\Sigma$.
The matrices weighting the proposal satisfy $A^2+B^2=I$ and take into account
appropriately the
covariance (likelihood) information. The proposal $Q(u,\cdot)$
is reversible with respect to $\mu := N(0,I)$, in the sense that
$$
\mu(du) Q(u,dv) = \mu(dv) Q(v,du).
$$
Note this implies $\int_E \mu(du) Q(u,dv) = \mu(dv)$.
The above proposal therefore provides an effective \emph{dimension-independent} (DI)
proposal for whitened Gaussian prior \mbox{$\eta_0=\mu$}, as the algorithm
is well-defined even on infinite-dimensional separable Hilbert spaces (i.e., even if $d=\infty$).
In the case of non-white priors,
e.g., for the standard assumption of a covariance which is trace-class operator, 
one must simply employ a change of variables. (See
\cite{cui2014dimension} for more details on this construction.)

{Recall the assumption from Section \ref{sec:inf} that the covariance is a
negative semi-definite perturbation of the prior,
so that if $\Sigma$ is the exact covariance, 
then $I_m - b_m \Sigma$ is guaranteed to
be positive semi-definite (and vanishing only off the true cLIS and
when $b_m=1$). 
When the approximate cLIS is constructed from samples in
practice one must take care to ensure non-negativity of
$I_m - b_m \Sigma$.
%

Note that as long as the proposal is split according to
a rotation induced by an operator $P$ with a 
finite range 
$m$,  then
{\em any proposal} can be used on the subspace spanned by
$P$ and the DI property will
be preserved.  However, the proposal  
should be chosen such that the resulting Metropolis-Hastings algorithm 
is convergent, as the above algorithm
is proven to be in \cite{rudolf2015generalization}.
If derivatives were available, we may use them on the cLIS part of
the proposal above to construct manifold-based proposals, as was recently
done in \cite{beskos2016geometric}.
The following proposal, which preserves the 
Gaussian approximation of the posterior on the cLIS
(instead of the prior) is not, in general,  geometrically ergodic
\begin{align*}
A & =  (1-b_m)^{1/2} P P^\top + (1- b_\perp^2)^{1/2} (I-PP^\top), \\
B & =  P \sqrt{b_m \Sigma} P^\top + b_\perp(I-PP^\top).
\label{eq:propa}
\end{align*}
In particular, it is shown in \cite{mengersen1996rates} Theorem 2.1
that this proposal is not ergodic for
$b=b_m=b_\perp=1$, for a wide range of target distributions,
including Gaussians with a covariance larger than $\Sigma$
on the 
subspace. 
This property is expected to hold for $b_m<1$ as well.
In \cite{chen2015accelerated} it is suggested simply to scale the covariance $\Sigma$
by a factor $(1+\epsilon)$, for $\epsilon>0$.
This strategy works in practice, but 
the downside is that there is no clear criterion for the choice of $\epsilon$.

\subsection{Multilevel cLIS}
\label{sec:mlis}

We will now embed the cLIS methodology within a multilevel sampling
framework. The idea here is that the cLIS is expected to converge at
some level of mesh refinement that is less accurate than the
final level required by the MLSMC algorithm, so that the cLIS can then
be embedded into higher levels at a nominal cost.
Furthermore, the telescopic identity can be leveraged along the way to
improve the cost of the algorithm.
Some justification/motivation for this idea lies in the typical structure of
Bayesian inverse problems: with a smoothing forward operator and/or
limited data, the likelihood-informed directions (the span of the cLIS)
tend to be relatively smooth. See \cite{cui2014likelihood} for an
example of the LIS basis converging under mesh refinement.

\subsubsection{Setting}
\label{sec:setting}

Recall that in the setting of Section \ref{sec:setup}, we are interested in a sequence
of unnormalized densities $\kappa_\ell(u_{0:\ell})$ in (\ref{eq:seq_kappa}) defined on spaces of increasing dimension $E_\ell$, for levels $\ell=0,1,\ldots$.
Let $h_\ell$ denote a resolution parameter and $\cC_\ell$  the associated computational cost  of evaluating $\kappa_\ell(u_{0:\ell})$,
such that $h_\ell \rightarrow 0$  and $\cC_\ell\rightarrow  \infty$ as $\ell \rightarrow \infty$,
and assume that the computational cost is dominated by a forward model $\cG_\ell(u_{0:\ell})$ involved in the likelihood calculation, as in \eqref{eq:seq_kappa}.
In particular, consider the case in which the sequence of spaces $E_0, E_1, ...., E_L \subset E$
correspond to finite-dimensional approximations (of increasing dimension) of a
limiting space $E:=E_\infty$,
where $E$ is a separable Hilbert space, and $u\in E$.

In order to establish a clear context, let
$\phi_1, \phi_2,\dots \in E$ and define
{${\Psi}_\ell := [\phi_1, \dots, \phi_{d_\ell}] \in E \times \bbR^{d_\ell}$}.
Using matrix notation, let
$$
E_\ell = (\Psi_\ell^\top \Psi_\ell)^{-1}\Psi_\ell^\top E.
$$
Letting 
$u_{0:\ell} = (\Psi_\ell^\top \Psi_\ell)^{-1}\Psi_\ell^\top u$,
then {$\Psi_\ell u_{0:\ell}$} is the orthogonal
projection of $u$ onto the $d_\ell$-dimensional
subspace of $E$ spanned by the columns of $\Psi_\ell$.
In the following $u_{0:\ell}$ may also correspond to the value of $u$ at $d_\ell$ grid points,
with $\Psi_\ell u_{0:\ell}$ an interpolant through those points.
In the limit, isomorphic representations of $E$ will be identified,
i.e., spatial representations or sequence representations in terms of expansion coefficients.
Suppose:
\begin{itemize}
\item{One has a regularly structured grid that is uniform across $D$.}
\item{An underlying spatio-temporal
dimension of the limiting space $E$, for example $L^2([0,1]^{D},\mathbb{R})$.} 
\item{The grid-spacing is $h_\ell$.}
\end{itemize}
Then 
the dimension of $E_\ell$ is $d_\ell = h_\ell^{-D}$. 
{Conversely, for an arbitrary expansion,
for example in terms of some family of orthonormal polynomials,
with equal numbers of basis functions in
each direction, it is reasonable to {\em define} $h_\ell := d_\ell^{-1/D}$.}
These notions are therefore interchangeable.
According to the simulated examples,
the cLIS associated with $E_\ell$ is expected to require $\cO(d_\ell)$ samples to identify.
Let $P_\ell \in \bbR^{d_\ell \times m_\ell}$ denote an orthonormal basis for the
$m_\ell-$dimensional cLIS at level $\ell$, so that
$$
C_\ell = I_{d_\ell} - Q_\ell Q_\ell^\top,
$$
where $Q_\ell = P_\ell \Lambda_\ell^{1/2}$,
for some diagonal matrix $\Lambda_\ell$ of non-zero singular values,
and $I_{d_\ell}$ is the ${d_\ell}\times{d_\ell}$ identity matrix.
We set $m = \lim_{\ell \rightarrow \infty} m_\ell$ and {let $P$ denote
the limiting $m-$dimensional cLIS on $E$.}
It is reasonable to assume that for $\ell$ sufficiently large, $m_\ell \approx m$,
{and ${\rm span} \Psi_\ell P_\ell$
will already provide a good approximation of $P$.}


The idea is that at some level $\ell^*$ in the MLSMC algorithm,
one stops constructing the cLIS
and the current cLIS $P_{\ell^*} \subset E_{\ell^*}$
is simply embedded into $E_{\ell*+n}$ for $n\geq 1$.
In this way, one can use the empirical covariance on the cLIS, at an $m-$dependent cost,
for a DILI proposal without recomputing the cLIS on higher levels.
{Henceforth the cLIS is constructed without reference to subsequent samples.
Therefore, within the MLSMC context,
one needs to collect at least $d_\ell$ samples for $\ell<\ell^*$,
but the restriction does not persist for $\ell>\ell^*$.
The implication of this is discussed in more detail in subsubsection \ref{sssec:mlcost2}.}

The rest of this subsection will be organized as follows.
The form of the importance sampling proposal will be described in Section \ref{sssec:propcon}.
The embedding of the cLIS will be described in Section \ref{sssec:embed},
and the multilevel covariance construction using the cLIS will be described in
Section \ref{sssec:mlco}.  The additional multilevel cost considerations
due to the DILI mutations are considered in Section \ref{sssec:mlcost2},
and finally an example of the framework for a simple basis is mentioned in
Section \ref{sssec:exam}.

%

\subsubsection{Importance Sampling Proposal To Extend Dimension}
\label{sssec:propcon}

The mutation kernel $K_\ell$ will be constructed through the DILI methodology in Section \ref{sec:dili}.
It  remains to determine the kernel $q_\ell:E_{\ell-1}\rightarrow\mathcal{P}(U_\ell)$ that extends
the dimension of the state space during the iterative importance sampling steps.
In both numerical applications in Section~\ref{sec:numerics} we employ regular grids
in 1D and 2D of increasing resolution; other options could involve truncating the Karhunen-Lo\`eve
expansion of the prior Gaussian measure.

In our applications we have used the Gaussian prior dynamics to determine $q_{\ell}$,
so that
\begin{equation*}
q_\ell(u_{0:\ell-1},du_\ell) = \mu_0( du_\ell |  u_{0:\ell-1} )\ ,
\end{equation*}
though other choices could also be made.
This choice gives  
$$G_\ell(u_{0:\ell}) = \cL_{\ell}(u_{0:\ell})/\cL_{\ell-1}(u_{0:\ell-1}).$$
From standard properties of Gaussian laws,
assuming that $\mu_0(du_{0:\ell})= N(0,\Gamma_{0:\ell})$ with covariance
\begin{equation*}
\Gamma_{0:\ell} = \left(
\begin{array}{cc}
\Gamma_{0:(\ell-1)}  & \Gamma_{0:(\ell-1),\ell} \\
 \Gamma_{0:(\ell-1),\ell}^{\top}  & \Gamma_{\ell,\ell}
\end{array}
\right) \, ,
\end{equation*}
%
with $\Gamma_{0:(\ell-1)} \in \bbR^{d_{\ell-1}\times d_{\ell-1}}$,
$\Gamma_{\ell,\ell}\in \bbR^{d'_\ell \times d'_\ell}$, $\Gamma_{0:(\ell-1),\ell}\in \bbR^{d_{\ell-1} \times d'_\ell}$,
we have
%
\begin{equation}\label{eq:klg}
q_\ell(u_{0:\ell-1},\cdot) =
\Gamma_{0:(\ell-1),\ell}^{\top}\,\Gamma_{0:(\ell-1)}^{-1} \, u_{0:\ell-1} +
N\big(0, \Gamma_{\ell} \big) \, ,
\end{equation}
where
$$
\Gamma_\ell := \Gamma_{\ell,\ell} - \Gamma_{0:(\ell-1),\ell}^{\top}\, \Gamma_{0:(\ell-1)}^{-1} \Gamma_{0:(\ell-1),\ell} \, .
$$
%


\subsubsection{cLIS Construction when Extending Dimension}
\label{sssec:embed}
Recall that the main idea in Section \ref{sec:setting} is that one will reach
a cut-off level, say $\ell-1$, when the standard cLIS methodology will be applied
using the particle information available at this point, as described in Section \ref{sec:dili},
but from level $\ell$ onwards  the cLIS  will simply be propagated forwards
without Monte Carlo effort to identify further directions informed by the likelihood.  We will
now describe how to carry out this propagation.

The construction of the cLIS proposal  in Section  \ref{sec:dili}
requires the identification of an orthonormal set of vectors spanning the critical
subspace informed by the likelihood after whitening the prior covariance. That is,  one must in practice work with the linear transformation $v_{0:\ell-1}=L_{\ell-1}^{-1} u_{0:\ell-1}$,
where $L_{\ell-1}$ is any matrix such that $L_{\ell-1} L_{\ell-1}^\top = \Gamma_{0:\ell-1}$.
Notice that in many cases (e.g.,\@ if the prior is a Gaussian Markov random field) $L_\ell^{-1}$ is sparse, so
this operation is cheap.
Assume that the columns of
matrix $P_{\ell-1}\in \mathbb{R}^{d_{\ell-1}\times m}$ correspond to the orthonormal basis of cLIS
at the cut-off level $\ell-1$, so that $P_{\ell-1}^{\top}P_{\ell-1}=\Lambda_{\ell-1}$ for a diagonal
matrix $\Lambda_{\ell-1}\in \mathbb{R}^{m\times m}$. We will identify $P_\ell$.

It will be convenient to define the matrices
\begin{equation} \label{eq:ay1}
A_{\ell | \ell-1} = \left(
\begin{array}{cc}
I_{d_{\ell-1}} \\
\Gamma_{0:(\ell-1),\ell}^\top \Gamma_{0:(\ell-1)}^{-1}
\end{array}
\right) \, , \qquad
A_{\ell \backslash \ell-1} = \left(
\begin{array}{cc}
 0_{d_{\ell-1}\times d_\ell'} \\
 \Gamma_{\ell}^{1/2}
\end{array}
\right) \, ,
\end{equation}
where $I_{d_{\ell-1}}$ is the $d_{\ell-1}\times d_{\ell-1}$ identity matrix and
$0_{d_{\ell-1}\times d_\ell'}$ is the $d_{\ell-1}\times d'_{\ell}$ matrix of all zeros.
Then for $u_\ell^* \sim q_\ell(u_{0:\ell-1},\cdot)$, one has%
\begin{equation*}
(u_{0:\ell-1}^{\top}, u_\ell^{*,\top})^\top = A_{\ell | \ell-1} u_{0:\ell-1} + A_{\ell \backslash \ell-1} \xi_\ell \, ,
\end{equation*}
where $\xi_\ell \sim N(0,I_{d'_\ell})$.
%
By definition, we have
\begin{equation}
\Gamma_{0:\ell} =  A_{\ell|\ell-1} \Gamma_{0:\ell-1} A_{\ell|\ell-1}^\top +
A_{\ell \backslash \ell-1}A_{\ell\backslash \ell-1}^\top.
\label{eq:prior1}
\end{equation}
It is less obvious that
\begin{equation}
\Gamma_{0:\ell-1}^{-1} =  A_{\ell|\ell-1}^\top \Gamma_{0:\ell}^{-1}A_{\ell|\ell-1}.
\label{eq:ortho}
\end{equation}
%
%
%
To see this, observe that
the first 
$d_{\ell-1}$ rows of $\Gamma_{0:\ell}$ are given by
$\Gamma_{0:\ell-1} A_{\ell | \ell-1}^\top$.
It is then 
clear that
\begin{equation*}
A_{\ell | \ell-1}^\top \Gamma_{0:\ell}^{-1} = \Gamma_{0:(\ell-1)}^{-1} \left(
\begin{array}{cc}
 I_{d_{\ell-1}}&  0_{d_{\ell-1} \times d_\ell'} \end{array}
\right) \, ,
\end{equation*}
and the second identity follows immediately.
Due to (\ref{eq:prior1}), (\ref{eq:ortho}) it is easy to check that
\begin{equation*}
L_{\ell}^{-1}\,\Gamma_{0:\ell}\,L_{\ell} = I_{d_\ell} - P_{\ell}\,\Lambda_{\ell}\,P_{\ell}^{\top}
\end{equation*}
where $L_\ell$ is such that $L_{\ell}L_{\ell}^{\top}=\Gamma_{0:\ell}$ and
\begin{equation*}
P_{\ell} = L_\ell^{-1}A_{\ell|\ell-1} L_{\ell-1} P_{\ell-1}\ , \quad P_{\ell}^{\top}P_{\ell}=: \Lambda_{\ell}
\equiv \Lambda_{\ell-1}\ .
\end{equation*}
%
This ensures that an orthonormal cLIS at the cut-off level $\ell-1$
transforms to an orthonormal cLIS at level $\ell$ through the map
$P_{\ell-1} \mapsto L_\ell^{-1}A_{\ell|\ell-1} L_{\ell-1} P_{\ell-1}$.

\subsubsection{Multilevel Covariance Estimation}
\label{sssec:mlco}

The 
covariance $C_\ell$ can also be estimated 
with the multilevel estimator \cite{bierig2015convergence, hoel2015multilevel}
\begin{equation}\label{eq:mlcov}
C^{\rm ML}_\ell \approx C_0^{N_0} + \sum_{l=1}^\ell \left( C^{N_l}_l - C^{N_l}_{l-1} \right),
\end{equation}
where $C^{N_l}_l = \frac1{N_l} \sum_{i=1}^{N_l} u_{0:l}^i(l) (u_{0:l}^i(l))^\top -
\left ( \frac1{N_l} \sum_{i=1}^{N_l} u_{0:l}^i(l) \right) \left ( \frac1{N_l} \sum_{i=1}^{N_l} u_{0:l}^i(l) \right)^\top$,
and $C^{N_l}_{l-1}$ is the appropriate upscaled (so that matrix dimensions match in \ref{eq:mlcov})) sample covariance associated with $u_{0:l-1}(l)$.
This will give rise to the multilevel cLIS approximation $P^{\rm ML}_{\ell}$,
which 
will be used to approximate the covariance on the approximate cLIS,
$\Sigma_{\ell} = P^{{\rm ML},\top}_{\ell} C_\ell P^{\rm ML}_{\ell}$, by
$$
\Sigma^{\rm ML}_{\ell} \approx P^{{\rm ML},\top}_{\ell} C_0^{N_0}P^{\rm ML}_{\ell} +
\sum_{l=1}^\ell P^{{\rm ML},\top}_{\ell}\left( C^{N_l}_l - C^{N_l}_{l-1} \right)P^{\rm ML}_{\ell}.
$$

Consider
$A_{\ell+n|\ell} = A_{\ell+n|\ell+n-1} A_{\ell+n-1|\ell+n-2} \cdots A_{\ell+1|\ell}$
for $A_{l|l-1}$, $l \ge 1$, defined in (\ref{eq:ay1}).
As mentioned in Section \ref{sssec:embed}, the cLIS will be constructed only until some level $\ell^*$,
after which it will be transformed to higher levels with the involvement  this operator.
The cLIS $P^{\rm ML}_{\ell^*}\in \bbR^{d_{\ell^*} \times m}$, constructed at the final level
using \eqref{eq:mlcov}, is transformed into the cLIS at higher levels
$P^{\rm ML}_{\ell^*+n}\in \bbR^{d_{\ell^*+n} \times m}$ by
\begin{equation}\label{eq:listrans}
P^{\rm ML}_{\ell^*+n} = L_{\ell^*+n}^{-1} A_{\ell^*+n|\ell^*} L_{\ell^*} P^{\rm ML}_{\ell^*},
\end{equation}
 where orthonormality of the column vectors of $P^{\rm ML}_{\ell^*+n}$ holds by transitivity and \eqref{eq:ortho}.
Recall that $L_\ell^{-1}$ is often sparse, for example for a Gaussian Markov random field,
and thus cheap to compute.
 Also, $L_\ell$ itself may be sparse, or have a simple structure which
 allows for cheap
 (i.e., not $\cO(d_n\ell^2)$) operations, as will be the case in Section \ref{ssec:diff} below.


\subsubsection{Multilevel Cost Considerations}
\label{sssec:mlcost2}

The multilevel analysis proceeds as in a standard case, except one has to consider that
if $\ell>\ell^*$ then $\cC_\ell = h_\ell^{-\gamma D}$  and
if $\ell\leq\ell^*$ then $\cC_\ell = h_\ell^{-3 D}$.
{Note that the cubic power corresponds to the worst case scenario for the cost of computing the cLIS,
while it may be possible in cases to compute
it more cheaply, e.g.,\@ even with linear cost}.
Assuming we fix $\ell^*$,
then asymptotically the use of the cLIS does not change the error analysis of the estimates.
One has $N_\ell = h_\ell^{(\beta+\gamma)/2}$ if $\ell>\ell^*$
and $N_\ell = h_\ell^{(\beta+3)/2}$ if $\ell\leq\ell^*$.
More careful analysis can be done, e.g.,\@ using the rate of convergence of the cLIS
to choose $\ell^*$, 
but since this construction is merely to improve mixing of the MCMC kernels,
it is reasonable to simply 
fix $\ell^*$ 
and ensure that $N_0$ is chosen large enough that $N_{\ell^*}>d_{\ell^*}$.

\subsubsection{Example with a Karhunen-Lo\`eve Basis}
\label{sssec:exam}

We end this subsection with a comment that
the spaces $\{E_\ell\}_{\ell=0}^L$ could be determined via a
Karhunen-Lo\`eve expansion as described below.

%
Let $\mu_0$ be the prior distribution 
over the infinite-dimensional separable  Hilbert space $E$, 
which will be assumed Gaussian with mean $0$ and trace-class covariance operator $\Gamma$.
There is an orthonormal basis $\{\phi_i\}_{i=1}^\infty$ for $E$ and
associated eigenvalues $\{\lambda_i\}_{i=1}^\infty$ such that $\Gamma\phi_i = \lambda_i \phi_i$.
The {\em Karhunen-Lo\`eve} expansion of a draw $u\sim \mu_0$ is given by
\begin{equation}\label{eq:kl}
u = \sum_{i=1}^\infty x_i \phi_i, \quad {\rm where} ~~
x_i = \langle u , \phi_i \rangle = \lambda_i^{1/2} \xi_i,
\quad {\rm and}~~ \xi_i\sim N(0,1) ~~ i.i.d..
\end{equation}
%
%
Thus, the covariance operator  $\Gamma$ is diagonal in the basis $\Psi_\infty=[\phi_1, \phi_2,\ldots]$.
In this setting it is natural to work with the coordinates $u_\ell = (x_{d_{\ell-1}+1}, \ldots, x_{d_{\ell}})$,
$0 \le \ell \le L$, so that we simply have
$L_\ell = \Gamma^{1/2}_{0:\ell} = \Psi_\ell^\top \Gamma^{1/2} \Psi_\ell = \mathrm{diag}\{\lambda_1^{1/2},\ldots,\lambda_{d_\ell}^{1/2}\}$.
Also, one has that
\begin{equation}\label{eq:klq}
q_\ell(u_{0:\ell-1},\cdot) = q_{\ell}(\cdot) =
N\left(0,{\rm diag}(\lambda_{d_{\ell-1}+1}, \lambda_{d_{\ell-1}+2},\dots, \lambda_{d_{\ell}}) \right) \, ,
\end{equation}
and $\Gamma_{0:(\ell-1),\ell}=0$ 
so
$P_{\ell^*+\ell} =
[P_{\ell^*}^\top, {\bf 0}_{ m_{\ell^*} \times (d_\ell-d_{\ell^*})}]^\top$,
where ${\bf 0}_{m\times n} \in \bbR^{m\times n}$ is a matrix of zeros.

\section{Examples}
\label{sec:numerics}


%
In this section, two models will be described and the assumptions on the selection functions
(A\ref{hyp:A}), (A\ref{hyp:C}) will be verified.
Before describing the examples in detail, we digress slightly to
discuss technical assumptions in Section \ref{ssec:restrict}.
Section \ref{ssec:diff} considers inversion of the white noise forcing in an
SDE given noisy observations of the path.
Section \ref{ssec:pde} considers a Bayesian inverse problem of inferring the
diffusion coefficient in a 2D elliptic PDE given noisy observations of
the solution field. 

\subsection{Restriction of prior measure}
\label{ssec:restrict}

In what follows, the prior measure $\mu_0$ will be Gaussian, and hence supported on
an unbounded space in principle.  The {\it restricted prior measure} is
\begin{equation*}
 \mu_{0,R}(du) := {\bf 1}_{S_R}(u)  \frac1{\mu_0(S_R)} \mu_0(du)\ , \quad
 S_R := \{u \in E ; 
 |u|_{L^\infty(\Omega)} \leq R\}\ ,
\end{equation*}
 for some $R>0$, where  $\Omega$ is the spatial/temporal domain.
 Note that provided $\mu_0(L^\infty(\Omega))=1$,
 for any $\varepsilon>0$, there exists a $R(\varepsilon)$ such
 that $|\mu_{0,R} - \mu_0|_{\rm TV} < \varepsilon$, as shown in
 \cite{rudolf2015generalization}.
 This restriction allows for a simple verification of assumptions
 (A\ref{hyp:A}) and  (A\ref{hyp:C}).
 In full generality one would have to carry out several technical proofs
that would obscure the main ideas of the ML approach.
It will be shown below that the restriction to  $S_R$
 will allow (A\ref{hyp:A}) and (A\ref{hyp:C}) to hold for the examples considered.
 Note that the bound on TV-norm implies a similar bound for the difference in
 expectation of bounded functionals, and functions with bounded second moments
 (via Hellinger metric, where the bound is replaced by $\varepsilon^{1/2}$,
 as shown in Lemma 1.30 of \cite{law2015data}).

Before continuing, assumption (A2) needs to be considered.
Theorem 20 of  \cite{rudolf2015generalization}  shows
 that  under conditions on the target distribution, the  Metropolis-Hastings algorithm with
 proposal~\eqref{eq:diliprop}
 restricted on $E_R$ has an $L^2(\mu_0)$ spectral gap
 (see also Corollary 4 of \cite{rudolf2015generalization} to verify that \eqref{eq:diliprop}, for $\bar u =0$,
 takes the appropriate so-called ``generalized pCN'' form).
 Therefore the proposal kernel $K_\ell$, conditionally on the current population of samples,
satisfies a spectral gap assumption.
It is beyond the scope of the present work to theoretically verify the validity of the algorithm for
this weaker property (relative to (A\ref{hyp:B})), so we shall content ourselves with the stronger assumption (A\ref{hyp:B}) and leave open the much more challenging
question of the algorithm's rigorous validity under weaker assumptions.
See also the recent work \cite{jasra17weak}
for consideration of weaker assumptions
in the case of the original MLSMC sampler algorithm on spaces of fixed dimensions of \cite{ourmlsmc}.

\subsection{Conditioned Diffusion}
\label{ssec:diff}

We consider an SDE scenario.
For $u$ denoting  a realisation of the $s$-dimensional Brownian motion, $s\ge 1$,
let $ p =p(u)$ be the solution of the SDE
\begin{equation}\label{eq:diff}
dp = f(p) dt + \sigma(p) d u, \quad p(0)=p_0,
\end{equation}
where $f: \bbR^s \mapsto \bbR^s$ and
$\sigma: \bbR^s \mapsto \bbR^{s\times s}$
are element-wise 
Lipschitz continuous,
with  $\sigma \in \bbR^{s\times s}$ non-degenerate.
For any $T>0$ 
there is a unique solution $p \in C(\Omega,\bbR^s)$ to \eqref{eq:diff},
with $\Omega=[0,T]$, and a map $u \mapsto p$
 continuous on $C(\Omega,\bbR^s) \mapsto C(\Omega,\bbR^s)$ with probability
 1 under the Wiener measure.
This is shown in Theorem 3.14 of \cite{hairer2011signal}, along
with the well-posedness of the corresponding  smoothing problem
below.
Note that since $C(\Omega,\bbR^s) \subset L^\infty(\Omega,\bbR^s)$
the prior Wiener measure can be restricted to some $S_R$ with arbitrarily small
effect.
%
%
Let $\cG_i(u) =  p(t_i;u)$, for times $0<t_1<\cdots <t_q \le T$, $q\ge 1$.
We consider observations $y|u \sim N(\cG(u),\Xi)$, where $\cG = (\cG_1^{\top},\ldots, \cG_q^{\top})^{\top}$,
with noise (of variance $\Xi$) independent of $u$,
so that the likelihood is  
%
\begin{equation}
\label{eq:like}
\cL(u) = \exp\left( -\tfrac12 |y-\cG(u)|_\Xi^2 \,\right).
\end{equation}

We will henceforth assume $s=1$, though multi-dimensional extensions are straightforward.
The standard Euler-Maruyama discretization is employed,
with refinement occurring via Brownian bridge sampling between
successive grid points; this is a particular scenario of the general description in Section \ref{sssec:propcon}.
The paths are generated on a uniform grid,
which gives rise to proposals of the form \eqref{eq:klg} under the prior Wiener measure dynamics.
In particular, let us assume that $d_\ell = d_0 \,2^{\ell}$,
so $h_\ell = T/d_\ell$; to avoid undue complications
$d_0$ is chosen large enough to accommodate the $q$
observations at grid points.
Then the linear transformations in \eqref{eq:klg}
are given for the case of the scalar $(s=1)$ SDE in \eqref{eq:diff}
by the following, for $i=1,\dots,d_{\ell}$
(the first, undefined, equation is ignored for $i=1$):
\begin{align*}
(A_{\ell+1|\ell})_{2i - 1, i-1} & =  (A_{\ell+1|\ell})_{2i - 1, i} = 1/2, \\
(A_{\ell+1|\ell})_{2i, 2i} & =  1, \\
(A_{\ell+1\backslash \ell})_{2i-1,i} & =  \sqrt{h_\ell}/2,
\end{align*}
and $(A_{\ell+1|\ell})_{j,k} = (A_{\ell+1\backslash \ell})_{j,k} = 0$ otherwise.
This is simply a way to write down the well-known Brownian
bridge measure for the fine grid points $u_{0:\ell+1}$---every other of which coincides
with one of the coarse grid points $u_{0:\ell}$, or bisects two of them.
The new bisecting points $u_{\ell+1}$ are conditionally independent given $u_{0:\ell}$,
with distribution,

$$
q_{\ell+1,i}(u_{0:\ell},u_{\ell+1,i}) =
N\left (\tfrac12(u_{0:\ell,i} + u_{0:\ell,i+1}), \tfrac{h_\ell}{4}\right),
$$
for $i=1,\dots, d_\ell$. Operator $L_\ell: v_{0:\ell} \mapsto u_{0:\ell}$ in \eqref{eq:ortho} is given by
the Cholesky factorization: $(L_\ell)_{j,i} = \sqrt{h_\ell}$, $i\leq
j$; $(L_\ell)_{j,i} = 0$ otherwise.

As described above, the path $p$ is a continuous 
function of the driving Brownian path $u$.  Likewise, the path
$p_{0:\ell}$ arising from the Euler-Maruyama discretization of \eqref{eq:diff}
using the Brownian motion positions $u_{0:\ell}$
is a continuous function of $u_{0:\ell}$.
The likelihood function at level $\ell$ will now be
$\cL_{\ell}(u_{0:\ell}) =  \exp\left( -\tfrac12 |y-\cG_\ell(u_{0:\ell})|_\Xi^2 \,\right)$,
with $\cG_\ell(u_{0:\ell})$ denoting the mapping from the Euler scheme points $u_{0:\ell}$
to the position of $p_{0:\ell}$ at observation times. We immediately have
that $|\cG_\ell(u_{0:\ell})| \le C(R)$, so assumption (A\ref{hyp:A}) is satisfied.
We also have
\begin{align*}
\left| \frac{\cL_{\ell}(u_{0:\ell})}{\cL_{\ell-1}(u_{0:\ell-1})} -1\right| &\le  C(R) \big|\cG_\ell(u_{0:\ell}) -\cG_{\ell-1}(u_{0:\ell-1})\big| \\  &\le C'(R) \sup_{t\in [0,T]}|p_\ell(t) - p_{\ell-1}(t)|  = o(h_{\ell-1}^{\beta/2})
\end{align*}
%
%
%
%
%
where $p_\ell(t) = p_{\ell,i}$ for $t\in [(i-1)h_\ell, i h_\ell)$,
the latter bound holding almost surely for any $\beta \in (0,1)$, as shown in Theorem 7.12 of \cite{graham2013stochastic}.
Note that this does not provide our required uniformity in $u_{0:\ell}$ for Assumption
(A\ref{hyp:C}); however, the required rate will be verified numerically.

The specific settings of our numerical study are as follows: $\sigma(p) =
1$, $T = q = 16$, and the observations are evenly spaced with $t_1 = 1$ and
noise $\Xi = 0.01$. The simulations are carried out with $d_0 = 32$ at the
initial level and $d_l = d_02^l$ as described above. 

Numerical results for solution of the conditioned diffusion problem are
shown in Figure \ref{fig:cond_diff}. The variance rate plot helps us to obtain
$\beta$ for our simulations. 
We then consider SMC (i.e.~no telescoping identity),
MLSMC with the standard pCN method for the mutations and MLSMC with the DILI
proposals of Section \ref{sec:mlsmcdili}. 
The samples for the simulations are chosen as mentioned above. The results are repeated 100 times and averaged for robustness.
The (theoretical) cost vs error plot of Figure \ref{fig:cond_diff}
presents a comparison of the three methods. Both MLSMC methods out perform SMC as
was the case in \cite{ourmlsmc}. Moreover, it is evident that the performance with the DILI mutations is superior
to that of the standard pCN mutations.


\begin{figure}[h]
\centering
\subfigure[Variance convergence rate.]
{\includegraphics[width=0.7\columnwidth,height=7cm]{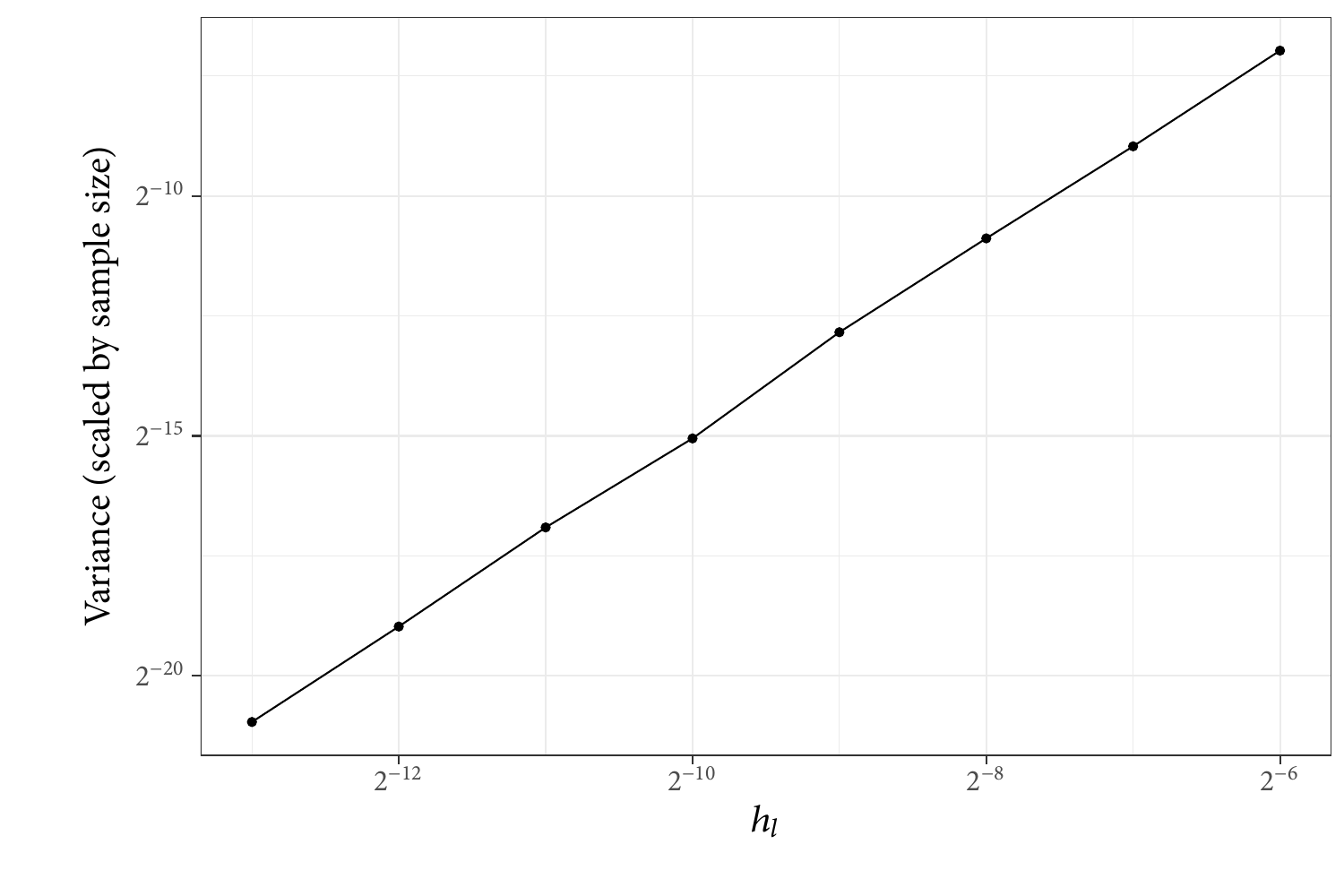}}
\subfigure[Cost vs. error.]
{\includegraphics[width=0.7\columnwidth,height=7cm]{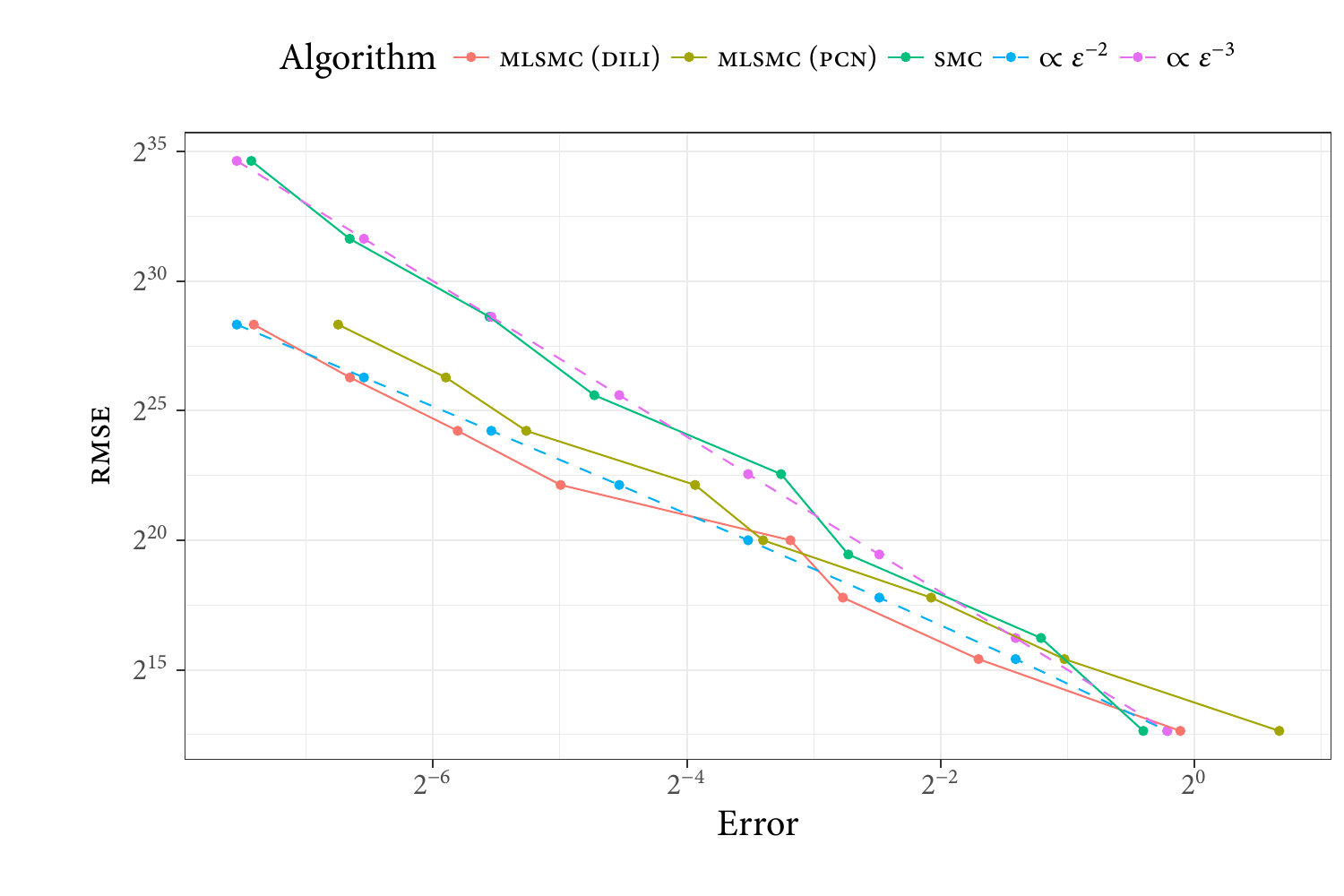}}
\caption{Results for the conditioned diffusion example.}
\label{fig:cond_diff}
\end{figure}


\subsection{Elliptic PDE Inverse Problem}
\label{ssec:pde}

In this section, we consider a Bayesian inverse problem involving
inference of the (log) permeability coefficient in a 2D elliptic PDE, given
noisy measurements of the associated solution field (representing,
e.g., pressure).
Consider the nested spaces
$V := H^1(\Omega) \subset L^2(\Omega) \subset H^{-1}(\Omega) =:V^{*}$,
for a domain
$\Omega \subset \bbR^D$ with convex boundary
$\partial \Omega \in C^0$.
Let $f \in V^{*}$, and consider the following PDE on $\Omega$
\begin{align}
  -\nabla \cdot (\bbK(u) \nabla p) &= f, \qquad\text{on }\Omega,
  \label{eq:elliptic} \\
  p &= 0, \qquad \text{on }  \partial \Omega.
  \label{eq:boundary}
\end{align}
for pressure field $p$, permeability
$\bbK(u) = e^u$, 
and known force vector field $f$.
We set up a Bayesian inference problem for the unknown log permeability field $u$.
We assume a truncated stationary Gaussian prior,
\begin{equation}
\label{eq:prior}
  u \sim \mu_0(du)\cdot \mathrm{I}\,[\,|u|_{\infty}<R\,]  , \quad
  \mu_0\equiv N(0,C),
\end{equation}
for some $R>0$, with $C$ denoting the covariance operator derived through
the covariance function
\begin{equation*}
c(x,x') = \sigma^2 \exp \{- |x-x'|^2/\alpha \},
\end{equation*}
for hyperparameters $\sigma>0$, $\alpha>0$

We will henceforth assume that $D=2$ and $\Omega=[0,1]^2$.
Let $p=p(\,\cdot\,;u)$ denote the weak solution of
\eqref{eq:elliptic}--\eqref{eq:boundary} for parameter $u$. Define the following
vector-valued function
\begin{equation*}
  \mathcal{G}(p) = [g_1(p),\dots,g_M(p)]^\top,
\end{equation*}
where $g_m$ are elements of the dual space $V^*$ for $m = 1,\dots,M$,
for some $M>1$. It is
assumed that the data take the form
\begin{equation}
  y = \mathcal{G}(p) + \nu, \qquad \nu \sim N(0,\Xi), \qquad \nu \perp u,
\end{equation}
so that the likelihood is given again by $\cL(u) = \exp( -\frac12 |y-\cG(p(\,\cdot\,;u))|_\Xi^2 )$.

Under the above choice of covariance function, $u$ is a.s.\@ continuous
on $\Omega$.
When $u\in L^\infty(\Omega)$, it is shown in
\cite{dashti2011uncertainty} that if  $\partial\Omega \in C^1$
then there exists a unique weak solution $p$ depending continuously on $u$,
and whose regularity is  determined by $f$. In particular, for given
$u\in L^\infty(\Omega)$,  if $f\in V^*$ then $p\in V$ while if
$f\in L^s(\Omega)$ for $s>1$ then $p\in L^\infty(\Omega)$.

The specific settings for our simulations and generated data are as
follows: the source/sink term
$f$ is defined by a superposition of four weighted Gaussian bumps with
standard deviation $\sigma_f = 0.05$, centered at $(0.3, 0.3)$, $(0.3,
0.7)$, $(0.7, 0.3)$ and $(0.7, 0.7)$, with weights $\{2, -3, -2,
3\}$, respectively. Observations of the potential function
$p$ are collected at 25 measurement points, evenly spaced within
$[0.2,
0.6]^2$ (boundaries included). The observation variance
$\sigma_y^2$ is chosen such that a prescribed signal-to-noise ratio
(SNR), which is defined as
$\max\{p\}/\sigma_y$, is equal to 10. The hyperparameters
$\alpha$ and
$\sigma^{-2}$ are given Gamma priors with mean and
variance $1$.

\subsubsection{Numerical Method and Multi-Level Approximation}

Consider the 1D piecewise linear nodal 
basis functions $\phi_j^1$ defined as follows,
for mesh $\{x_i = i / K \}_{i=0}^K$
\[
\phi^1_j(x) =
\begin{cases}
\frac{x-x_{j-1}}{x_j-x_{j-1}}, & x\in [x_{j-1},x_j] \\
1- \frac{x-x_{j}}{x_{j+1}-x_{j}}, & x\in [x_{j},x_{j+1}] \\
0, & {\rm else} \ .
\end{cases}
\]
Now consider the tensor product grid over $\Omega=[0,1]^2$
formed by $\{(x_i,x_j)\}_{i,j=1}^K$, where $K=K_0\times2^{\ell}$, 
where the initial resolution $K_0 = 10$.
Let $\phi_{i,j}(x,y) = \phi^1_j(x) \phi^1_i(y)$ be piecewise bilinear functions,
and let $E_\ell = {\rm span}\{ \phi_i \}_{i=1}^{d_\ell}$, with $d_\ell=K^2$,
and any appropriate single index representation.
The permeability at level $\ell$ will be approximated by
$\bbK_{\ell}(u_{0:\ell}) = \sum_{i=1}^{d_\ell} e^{u_{0:\ell}^i} \phi_i$.
Likewise, the solution will be approximated by
$p_\ell(u_{0:\ell}) = \sum_{i=1}^{d_\ell} p_{\ell}^i \phi_i$.
The weak solution of the PDE
(\ref{eq:elliptic})--(\ref{eq:boundary}) is generated by a standard finite element
approximation, resulting in the 
solution ${\bf p} := p_\ell(u_{0:\ell})$.
This is done by substituting these expansions into \eqref{eq:elliptic}
and taking inner product with $\phi_j$ for $j=1,\dots, d_\ell$.
Define ${\bf f}_j = \langle f, \phi_j \rangle$ and
$$
{\bf A}_{ij} := \sum_{k_1=j_1-1}^{j_1+1} \sum_{k_2=j_2-1}^{j_2+1}
\int_{x_{j_1-1}}^{x_{j_1+1}} \int_{y_{j_2-1}}^{y_{j_2+1}}
e^{u_{0:\ell}^k} \phi_k \nabla\phi_i \cdot \nabla \phi_j dx dy ,
$$
where the notation $j=(j_1,j_2)$ is introduced to
represent the components of the indices corresponding to spatial dimensions $1$ and $2$.
The approximate weak solution to equations \eqref{eq:elliptic}, \eqref{eq:boundary}
is given by the system ${\bf A  p} = {\bf f}$.

The solution $p_\ell(u_{0:\ell})$ is then plugged into the likelihood to provide
$\mathcal{L}_\ell(u_{0:\ell})$.
At the next level, values of log-permeability on extra grid points are proposed from the conditional
prior dynamics $u_{\ell+1}|u_{0:\ell}$, by halving horizontal/vertical
distances between points in the grid.

Notice that in the case $R\rightarrow \infty$, 
$\cL_{\ell}(u_{0:\ell})$, is	 not {\em  uniformly bounded}
for the full unrestricted support of the  Gaussian measure $\mu_0$.
Choosing $R<\infty$, the weak form of the equation \eqref{eq:elliptic}
is continuous and coercive uniformly in $u$,
and Lax-Milgram Lemma holds \cite{ciarlet2002finite}.
This provides the uniform bound in (A\ref{hyp:A}).
Uniform bounds on the PDE finite-element approximations
with piecewise bilinear nodal basis functions are readily available
in this case for any fixed space $E_\ell$.
See \cite{ourmlsmc, ciarlet2002finite, zienkiewicz1977finite} for details.


Now, we proceed to extend the proof of convergence rate from finite uniform $u$ \cite{ourmlsmc}
to infinite (truncated) Gaussian $u$.
We define the $V$-norm as
\begin{equation*}
|p|^2_V := 
\int_{[0,1]^2} |\nabla p(u)|^2  dx, \quad p\in V \, ,
\end{equation*}
noting that the boundary condition \eqref{eq:boundary}
guarantees that $\int_{\Omega} p dx = 0$ and so Poincar\'e inequality applies.
As in \cite{ourmlsmc}, the quantity we would like to
bound uniformly in $u$ is
\begin{equation}\label{eq:tri}
|p_\ell(u_{0:\ell}) - p(u)|_V \leq |p_\ell(u_{0:\ell}) - p(u_{0:\ell})|_V
+ |p(u_{0:\ell}) - p(u)|_V.
\end{equation}
The first term is dealt with as in \cite{ourmlsmc}.
The second term comes from the truncation to $E_\ell$.
Denote $\bar p = p(u_{0:\ell})$
and observe that for all $v\in V$
\begin{equation*}
\langle \nabla v , \bbK_\ell(u_{0:\ell}) \nabla \bar p - \bbK(u) \nabla p \rangle = 0,
\end{equation*}
so
\begin{equation*}
\langle \nabla v , \bbK_\ell(u_{0:\ell}) \nabla \bar p - \bbK_\ell(u_{0:\ell}) \nabla p \rangle
+ \langle \nabla v , \bbK_\ell(u_{0:\ell}) \nabla p - \bbK(u) \nabla p \rangle = 0.
\end{equation*}
Letting $v=\bar p - p$ and rearranging, we have (using also Poincar\'{e} inequality)
\begin{align*}
| \bar p - p |_V^2 & \leq  C(|u_{0:\ell}|_{\infty})
| (\bbK_\ell(u_{0:\ell})-\bbK(u)) (\nabla p)\cdot (\nabla (\bar p - p)) | \\
& \leq  C(|u_{0:\ell}|_{\infty})
| \bbK_\ell(u_{0:\ell})-\bbK(u)|_{L^\infty(\Omega)} |  p |_V | \bar p - p |_V
\end{align*}
Therefore on the truncated space $E_R$, the following holds
\begin{equation}\label{eq:desid}
| \bar p - p |_V \leq C(R) | \bbK_\ell(u_{0:\ell})-\bbK(u) |_{L^\infty(\Omega)}
= \cO(h_\ell^{\beta/2}) \, , 
\end{equation}
%
%
for some $\beta\in (2,4)$ (see Section 3.3 of \cite{ciarlet2002finite}).
The error due to the solution of the PDE with
finite element discretization of diameter
$h_{\ell}$ is also given by
\begin{equation*}
|p_\ell(u_{0:\ell}) - p(u_{0:\ell})|_V = \cO(h_\ell^{\beta/2}),
\end{equation*}
for $\beta \in (2,4)$
\cite{ourmlsmc, ciarlet2002finite}. 
Ultimately, the quantity
\begin{equation*}
V_\ell= \max \{\|G_\ell - 1\|^2_\infty, \|\rho_\ell-\rho_{\ell-1}\|^2_\infty\}
\end{equation*}
can be bounded by
$C h_\ell^\beta$,
as both terms are controlled by \eqref{eq:tri}.
The first term is handled similarly to the work \cite{ourmlsmc}.
Typically the functions $\rho_\ell$ 
we are interested in
will have the form $\rho_\ell(u_{0:\ell})_i = f_i(p_\ell(u_{0:\ell}))$ for some $f_i\in V^*$,
Hence $V_\ell = \cO(\||p_\ell(u_{0:\ell}) - p(u)|_V\|_\infty)$.

%



Numerical results for the elliptic PDE inverse problem are presented
in Figure \ref{fig:pde}, which contains plots analogous to those shown
for the previous numerical example. Again, the MLSMC schemes show the
desired improved convergence rate, and the DILI mutation steps yield
consistently better performance than pCN mutations.

%
%



\begin{figure}[h]
\centering
\subfigure[Variance convergence rate.]
{\includegraphics[width=0.7\columnwidth,height=7cm]{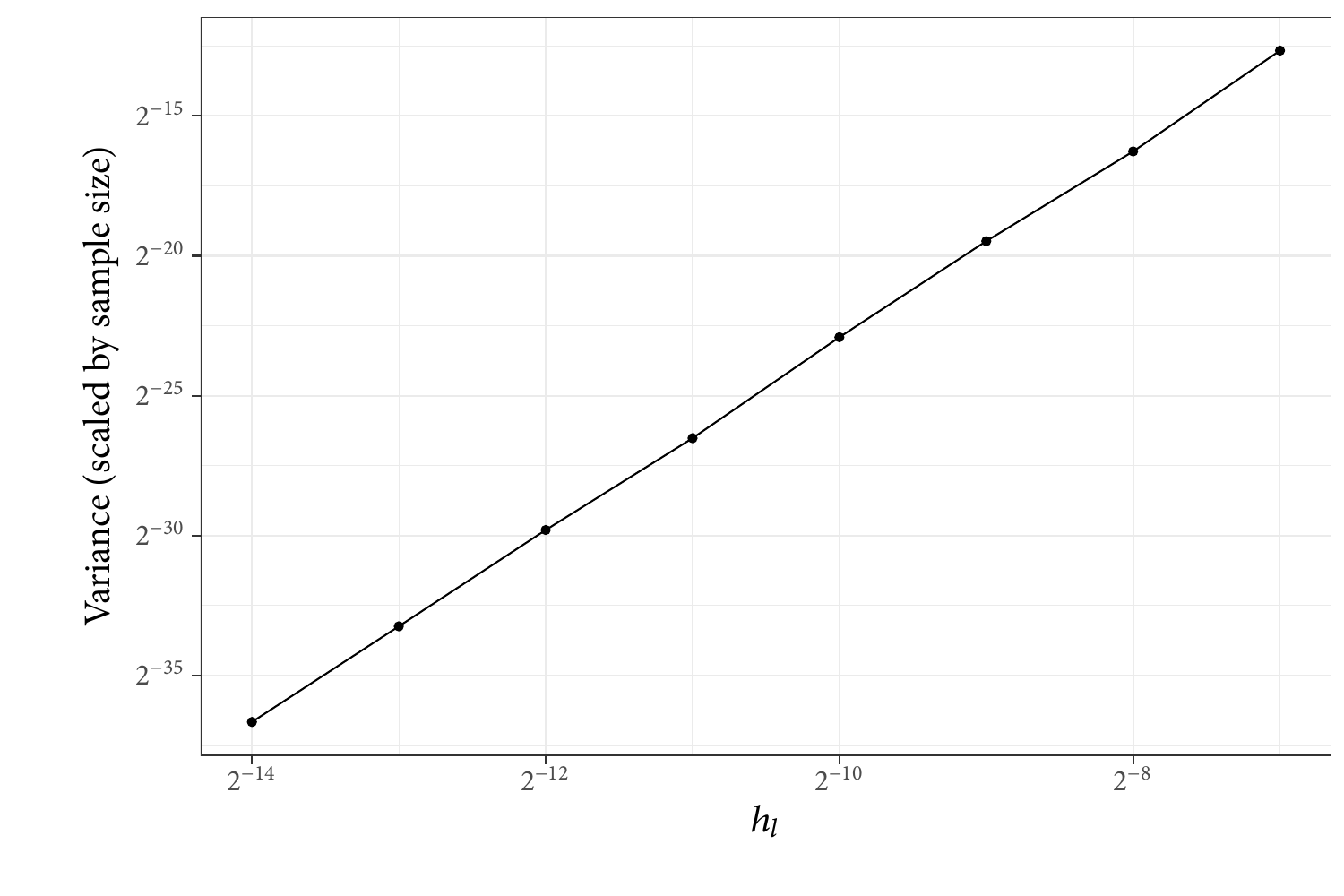}}
\subfigure[Cost vs. error.]
{\includegraphics[width=0.7\columnwidth,height=7cm]{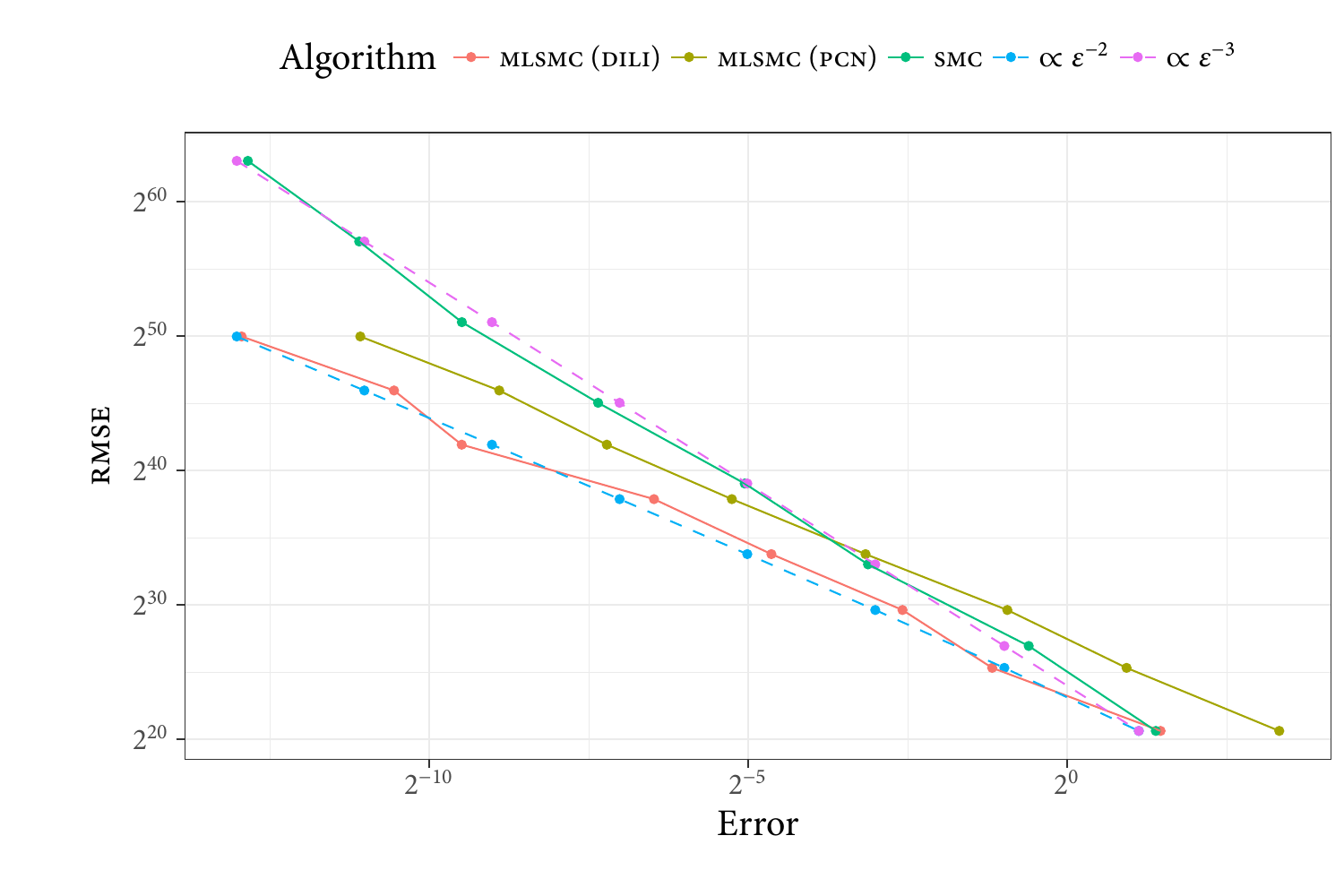}}
\caption{Results for the 2D PDE example.}
\label{fig:pde}
\end{figure}

\subsubsection*{Acknowledgements}
AJ \& YZ were supported by Ministry of Education AcRF tier 2 grant,
R-155-000-161-112.  KJHL was supported by an ORNL LDRD Strategic Hire
grant number 32112580 and also in part by the DOE DiaMonD MMICC. YM was
supported in part by the US Department of Energy (DOE) Office of
Advanced Scientific Computing Research under grant number DE-SC0009297
(DiaMond MMICC).  AB was supported by the Leverhulme Trust Prize.
\appendix

\section{Technical Result}

The following lemma is similar to Theorem 3.1 in \cite{ourmlsmc},
and the proof follows in the same spirit, but is given for completeness.

\begin{lemma}\label{lem:main}
Assume (A\ref{hyp:A}-\ref{hyp:C}).
 Then there exists a $C>0$ and  $\kappa\in (0,1)$ such  that for any
 $g\in\mathcal{B}_b(E)$, with $\|g\|_{\infty}=1$,
\begin{align*}
\mathbb{E}\big[\{\hat \eta^{\rm ML}_L(g)-\mathbb{E}_{\eta_L}[g(U)]\}^2\big]
\leq
C\,\bigg(\frac{1}{N_0} + &\sum_{l=1}^{L}\frac{V_l}{N_{l}} 
+ \sum_{1\le l<q\le L}V_l^{1/2} V_q^{1/2}
\big\{\tfrac{\kappa^{q}}{N_{l}}
+\tfrac{1}{N_{l}^{1/2}N_{q}}
\big\} \bigg)\ .
\end{align*}
\end{lemma}

\begin{proof}

 The proof follows essentially that of \cite{ourmlsmc} given the above assumptions.
Assumptions (A\ref{hyp:A}-\ref{hyp:B}) are the similar to that paper.
Note that, as shown in Section 4.2 of \cite{ourmlsmc}, there is a constant $C>0$ such that
\begin{equation}\label{eq:gzg}
\|\tfrac{Z_{l-1}}{Z_{l}}G_{l}-1\|_{\infty} \leq C \|G_{l}-1\|_{\infty}.
\end{equation}

Observe that
$\hat\eta_l(\varphi\circ \rho_l) = \eta_l(G_l \varphi \circ \rho_l)/\eta_l(G_l)$.
Now establish the following notations
\begin{align}
Y_{l}^{N_{l}} &= \frac{\eta_{l}^{N_{l}}(G_{l} \varphi \circ \rho_l)}{\eta_{l}^{N_{l}}(G_{l})}
- \eta_{l}^{N_{l}}(\varphi \circ \rho_{l-1})\ , \quad  \nonumber \\[0.2cm]
Y_{l} &= \frac{\eta_{l}(G_{l}\varphi \circ \rho_l)}{\eta_{l}(G_{l})} - \eta_{l}(\varphi\circ \rho_{l-1})
\,\,\,\,\big(\,  \equiv \eta_{l}(g) - \eta_{l-1}(g)\, \big)\ , \label{eq:analytical} \\[0.3cm]
\nonumber
\overline{\varphi_l}(u) & = \big(\tfrac{Z_{l-1}}{Z_l}G_{l}(u)-1\big) \ , \\[0.3cm]
\nonumber
\widetilde{\varphi_l}(u)
&= \overline{\varphi_l}(u) \varphi(u)
\ , \\[0.3cm]
\label{eq:ay}
A_n(\varphi,N) &  = \eta_n^N( G_n \varphi \circ \rho_n)/\eta_n^N(G_n) \ ,
\quad \varphi\in\mathcal{B}_b(E)\ , \quad
0\leq n\leq L-1  \ , \\[0.2cm]
\label{eq:aybar}
\overline{A}_n(\varphi,N) &  =  A_n(\varphi,N) - \frac{\eta_n(G_n \varphi \circ \rho_n)}{\eta_n(G_n)}\ .
\end{align}

Notice that $\eta_l(\overline{\varphi_l}) = 0$ and $\eta_l(G_l)=Z_l/Z_{l-1}$.  So,
\begin{equation}\label{eq:newexpand}
Y_{l}^{N_{l}}-Y_{l} =
\underbrace{A_{l}(\varphi,N_{l})\,\{\eta_{l} - \eta_{l}^{N_{l}}\} (\overline{\varphi_l})}_{T^1_l}
+ \underbrace{\{\eta_{l}^{N_{l}}-\eta_{l}\} (\widetilde{\varphi_l}\circ \rho_l)}_{T^2_l}   +
 \underbrace{\{\eta_{l}^{N_{l}}-\eta_{l}\} (\varphi \circ (\rho_l - \rho_{l-1}))}_{T^3_l}  \ .
\end{equation}
Observe that there is an additional term $T_l^3$
in comparison to Eq. (10) of \cite{ourmlsmc}.
Lemma 3.1 of that paper is replaced by
\begin{equation}\label{eq:newdiagbound}
\|Y_{l}^{N_{l}}-Y_{l}\|_2^2 \leq
4 \|A_{l}(\varphi,N_{l})\,\{\eta_{l} - \eta_{l}^{N_{l}}\} (\overline{\varphi_l}) \|_2^2
+ 4\|\{\eta_{l}^{N_{l}}-\eta_{l}\} (\widetilde{\varphi_l}\circ \rho_l)\|^2_2 +
4\| \{\eta_{l}^{N_{l}}-\eta_{l}\} (\varphi \circ (\rho_l - \rho_{l-1}))\|^2_2 \ .
\end{equation}

In view of \eqref{eq:gzg} and \cite[Theorem 7.4.4]{delm:04},
the first 2 terms are bounded by $C \| G_l - 1 \|^2_\infty / N_l$
and the last term is bounded by $C \| \rho_l - \rho_{l-1} \|^2_\infty / N_l$.
Now
\begin{equation*}
\mathbb{E}\Big[\big\{\sum_{l=1}^{N}(Y_{l}^{N_{l}} - Y_{l})\big\}^2\Big]
= \mathbb{E}\Big[ \sum_{l=1}^{N}(Y_{l}^{N_{l}} - Y_{l})^2\Big] +
2\sum_{1\le l < q\le L}  \mathbb{E}\big[(Y_{l}^{N_{l}} - Y_{l}) (Y_{q}^{N_{q}} - Y_{q})\big],
\end{equation*}
and the cross terms are
\begin{align*}
\sum_{1\le l < q\le L} \mathbb{E} \big[(Y_{l}^{N_{l}} - Y_{l}) (Y_{q}^{N_{q}} - Y_{q})\big]  =   
 & \sum_{1\le l < q\le L}  \mathbb{E}  (T^1_lT_q^1) & (a)  \\
 & +  \sum_{1\le l < q\le L} \mathbb{E} (T^1_lT_q^2) +
\mathbb{E} (T^1_lT_q^3) & (b) \\
& + \sum_{1\le l < q\le L}
\mathbb{E} (T^2_lT_q^1) +  \mathbb{E} (T^3_lT_q^1)  & (c) \\
&  + \sum_{1\le l < q\le L} \mathbb{E} (T^2_lT_q^2) + \mathbb{E} (T^2_lT_q^3)
+  \mathbb{E} (T^3_lT_q^2) + \mathbb{E} (T^3_lT_q^3). & (d)
\end{align*}
There are 5 new terms with respect to \cite{ourmlsmc} (all those including $T^3$),
i.e., 1 in (b), 1 in (c), and 3 in (d), but they can be dealt with similarly.  In fact, since
$\|\tilde{\varphi_n}\|_\infty \leq \|\overline{\varphi_n}\|_\infty \leq C \|G_n - 1\|_\infty$,
and $\max \{\|G_n - 1\|^2_\infty, \|\rho_n-\rho_{n-1}\|^2_\infty\} = V_n$, the terms
are all of the same type as in \cite{ourmlsmc}, grouped by category (a,b,c,d),
and are bounded exactly as in the appendix of that paper.
\end{proof}

\end{document}